\documentclass[a4paper]{article}
\usepackage{geometry}
\usepackage[utf8]{inputenc}
\usepackage[english]{babel}

\usepackage{microtype}
\usepackage{amssymb,amsthm,amsmath}
\usepackage{xspace}
\usepackage{cite}
\usepackage{graphicx}
\usepackage{enumerate}
\usepackage{enumitem}
\usepackage{url}
\usepackage[hyperfootnotes=false]{hyperref}
\usepackage{textcomp}
\usepackage{titling}
\usepackage{float}
\usepackage{caption}
\usepackage{tikz}
\tikzset{ %
	V/.style={circle,draw=black,inner sep=0pt,minimum size=8pt,fill=gray!20},
	A/.style={V,fill=blue!20},
	B/.style={V,fill=red!80!yellow!50!gray!90}
}

\theoremstyle{plain}
\newtheorem{theorem}{Theorem}
\newtheorem{corollary}[theorem]
{Corollary}
\newtheorem{lemma}[theorem]
{Lemma}
\newtheorem{observation}[theorem]
{Observation}
\newtheorem{definition}[theorem]
{Definition}

\newtheorem{claim}[theorem]
{Claim}

\newcommand{\aic}{{\bf aic}}

\newcommand{\ctw}{\ensuremath{\mathbf{cw}}\xspace}

\newcommand{\dcw}{{\bf dcw}}

\newcommand{\Oh}{{O}}

\newcommand{\NN}{\ensuremath{\mathbb{N}}}

\newcommand{\cuts}{\mathtt{cuts}}
\newcommand{\width}{\mathtt{width}}

\newcommand{\obs}{\ensuremath{\mathbf{obs}}}
\newcommand{\C}{\ensuremath{\mathcal{C}}}

\newcommand{\intv}[2]{[#1,#2]}

\def\cqedsymbol{\ifmmode$\lrcorner$\else{\unskip\nobreak\hfil
\penalty50\hskip1em\null\nobreak\hfil$\lrcorner$
\parfillskip=0pt\finalhyphendemerits=0\endgraf}\fi} 
\newcommand{\cqed}{\renewcommand{\qed}{\cqedsymbol}}

\begin{document}

\thanksmarkseries{alph}

\title{{\bf\Large  Cutwidth: obstructions and algorithmic aspects}%
\thanks{
This work was partially done while Archontia C. Giannopoulou was holding a post-doc position at Warsaw Center of Mathematics and Computer Science.
The research of Archontia C. Giannopoulou has been supported by the European Research Council (ERC) under the European Union's Horizon
2020 research and innovation programme (ERC consolidator grant DISTRUCT, agreement No 648527).
The research of Michał Pilipczuk and Marcin Wrochna is supported by the Polish National Science Center grant SONATA UMO-2013/11/D/ST6/03073.
The research of Jean-Florent Raymond is supported by the Polish National Science Center grant PRELUDIUM UMO-2013/11/N/ST6/02706.
Michał Pilipczuk is supported by the Foundation for Polish Science (FNP) via the START stipend programme.
} $^{,}$\thanks{Emails:
\nolinkurl{archontia.giannopoulou@tu-berlin.de}, 
\nolinkurl{michal.pilipczuk@mimuw.edu.pl}, 
\nolinkurl{jean-florent.raymond@mimuw.edu.pl},
\nolinkurl{sedthilk@thilikos.info}, and
\nolinkurl{m.wrochna@mimuw.edu.pl}.
}
}

\author{Archontia C. Giannopoulou\thanks{Technische Universität Berlin, Berlin, Germany.}
\and Michał Pilipczuk\thanks{Institute of Informatics, University of Warsaw, Poland.}
\and Jean-Florent Raymond$^{\mbox{\footnotesize d}}$$^,$\thanks{AlGCo project team, CNRS, LIRMM, Montpellier, France.}\vspace{3mm}
\and Dimitrios  M. Thilikos$^{\mbox{\footnotesize e}}$$^,$\thanks{Department of Mathematics, National and Kapodistrian University of Athens, Athens, Greece.}
\and Marcin Wrochna$^{\mbox{\footnotesize d}}$
}

\date{\empty}

\maketitle

\vspace{-9mm}
\begin{abstract}
\noindent Cutwidth is one of the classic layout parameters for graphs. 
It measures how well one can order the vertices of a graph in a linear manner, so that the maximum number of edges between any prefix and its complement suffix is minimized.
As graphs of cutwidth at most $k$ are closed under taking immersions, the results of Robertson and Seymour imply that 
there is a finite list of minimal immersion obstructions for admitting a cut layout of width at most $k$.
We prove that every minimal immersion obstruction for cutwidth at most $k$ has size at most $2^{\Oh(k^3\log k)}$.
%
%

As an interesting algorithmic byproduct, we design a new fixed-parameter algorithm for computing the cutwidth of a graph that runs in time $2^{\Oh(k^2\log k)}\cdot n$,
where $k$ is the optimum width and $n$ is the number of vertices. 
While being slower by a $\log k$-factor in the exponent than the fastest known algorithm, given by Thilikos, Bodlaender, and Serna in [Cutwidth {I:} {A} linear time fixed parameter algorithm, {\em J. Algorithms}, 56(1):1--24, 2005] and [Cutwidth {II:} {A}lgorithms for partial $w$-trees of bounded degree, \newblock {\em J. Algorithms}, 56(1):25--49, 2005], 
our algorithm has the advantage of being simpler and self-contained; arguably, it explains better the combinatorics of optimum-width layouts.
\end{abstract}

\vspace{3mm}


\vspace{6mm}

\section{Introduction}\label{sec:intro}

The cutwidth of a graph is defined as the minimum possible {\em{width}} of a linear ordering of its vertices, where the width of an ordering $\sigma$ is the maximum, among all the prefixes of $\sigma$,
of the number of edges that have exactly one vertex in a prefix.
Due to its natural definition, cutwidth has various applications in a range of practical fields of computer science: 
whenever data is expected to be roughly linearly ordered and dependencies or connections are local, one can expect the cutwidth of the corresponding graph to be small.
These applications include circuit design, graph drawing, bioinformatics, and text information retrieval; we refer to the survey of layout parameters of  Díaz,  Petit, and  Serna~\cite{DiazPS02} for a broader discussion.

As finding a layout of optimum width is NP-hard~\cite{garey1979computers}, the algorithmic and combinatorial aspects of cutwidth were intensively studied.
There is a broad range of polynomial-time algorithms for special graph classes~\cite{HeggernesLMP11,HeggernesHLN12,Yannakakis85}, approximation algorithms~\cite{LeightonR99}, and fixed-parameter algorithms~\cite{ThilikosSB05,ThilikosSB05a}.
In particular, Thilikos, Bodlaender, and Serna~\cite{ThilikosSB05,ThilikosSB05a} proposed a fixed-parameter algorithm for computing the cutwidth of a graph that runs\footnote{Thilikos, Bodlaender, and Serna~\cite{ThilikosSB05,ThilikosSB05a} do not
specify  the parametric dependence of the running time of their algorithm. A careful analysis of their algorithm yields the above claimed running time bound.}
in time~$2^{\Oh(k^2)}\cdot n$, 
where $k$ is the optimum width and $n$ is the number of vertices. Their approach is to first compute the pathwidth of the input graph, which is never larger than the cutwidth.
Then, the optimum layout can be constructed by an elaborate dynamic programming procedure on the obtained path decomposition. 
To upper bound the number of relevant states, the authors had to understand how an optimum layout can look in a given path decomposition.
For this, they borrow the technique of {\em{typical sequences}} of Bodlaender and Kloks~\cite{BodlaenderK96}, which was introduced for a similar reason, but for pathwidth and treewidth instead of cutwidth.

Since the class of graphs of cutwidth at most $k$ is closed under immersions, 
and the immersion order is a well-quasi ordering of graphs\footnote{All graphs considered in this paper may have parallel edges, but no loops.}~\cite{RobertsonS10},
it follows that for each $k$ there exists a {\sl finite} obstruction set $\mathcal{L}_k$ of graphs such that a graph has cutwidth at most $k$ if and only if it does not admit any graph from $\mathcal{L}_k$ as an immersion.
However, this existential result does not give any hint on how to generate, or at least estimate the sizes of the obstructions.
The sizes of obstructions are important for efficient treatment of graphs of small cutwidth; this applies also in practice, as indicated by Booth et al.~\cite{BoothGLR92} in the context of VLSI design.

The estimation of sizes of minimal obstructions for graph parameters like pathwidth, treewidth, or cutwidth, has been studied before.
For minor-closed parameters pathwidth and treewidth, 
Lagergren~\cite{Lagergren98} showed that any minimal minor obstruction to admitting a path decomposition of width $k$ has size at most single-exponential in $\Oh(k^4)$,
whereas for tree decompositions he showed an upper bound double-exponential in $\Oh(k^5)$ .
Less is known about immersion-closed parameters, like cutwidth. 
Govindan and Ramachandramurthi~\cite{GOVINDAN2001189} showed that the number of minimal immersion obstructions for the class of graphs of cutwidth at most $k$ is at least $3^{k-7}+1$,
and their construction actually exemplify minimal obstructions for cutwidth at most $k$ with ${(3^{k-5}-1)}/{2}$ vertices.
To the best of our knowledge, nothing was known about upper bounds for the cutwidth case.

\subsection{Results on obstructions.} Our main result concerns the sizes of obstructions for cutwidth.

\begin{theorem}\label{thm:main}
Suppose a graph $G$ has cutwidth larger than $k$, but every graph with fewer vertices or edges (strongly) immersed in $G$ has cutwidth at most $k$.
Then $G$ has at most $2^{\Oh(k^3\log k)}$ vertices and edges.
\end{theorem}

\noindent
The above result immediately gives the same upper bound on the sizes of graphs from the minimal obstruction sets $\mathcal{L}_k$ 
as they satisfy the prerequisites of Theorem~\ref{thm:main}.
This somewhat matches the $({3^{k-5}-1})/{2}$ lower bound of Govindan and Ramachandramurthi~\cite{GOVINDAN2001189}.

Our approach for Theorem~\ref{thm:main} follows the technique used by Lagergren~\cite{Lagergren98} to prove that minimal minor obstructions for pathwidth at most $k$ have sizes single-exponential in $\Oh(k^4)$.
Intuitively, the idea of Lagergren is to take an optimum decomposition for a minimal obstruction, which must have width $k+1$, and to assign to each prefix of the decomposition one of finitely many ``types'',
so that two prefixes with the same type ``behave'' in the same manner. If there were two prefixes, one being shorter than the other, with the same type, then one could replace one with the other, thus
obtaining a smaller obstruction. Hence, the upper bound on the number of types, being double-exponential in $\Oh(k^4)$, gives some upper bound on the size of a minimal obstruction.
This upper bound can be further improved to single-exponential by observing that types are ordered by a natural domination relation, and the shorter a prefix is, the weaker is its type.
An important detail is that one needs to make sure that the replacement can be modeled by minor operations. 
For this, Lagergren uses the notion of {\em{linked path decompositions}} (a weaker variant of {\em{lean path decompositions}}; cf.~\cite{Thomas90,BellenbaumD02}).

To prove Theorem~\ref{thm:main}, we perform a similar analysis of prefixes of an optimum ordering of a minimal obstruction. 
We show that prefixes can be categorized into a bounded number of types, each comprising prefixes that have the same ``behavior''.
Provided two prefixes with equally strong type appear one after the other, we can ``unpump'' the part of the graph in their difference.

To make sure that unpumping is modeled by taking an immersion, we define {\em{linked orderings}} for cutwidth and reprove the analogue of the result of Thomas~\cite{Thomas90} (see~\cite{BellenbaumD02} for simplified proofs): 
there is always an optimum-width ordering that is linked.
We remark this already follows from more general results on submodular functions: the same is true for parameters like \emph{linear rank-width}, as observed by Kant\'{e} and Kwon~\cite{KanteK14}, which in turns follows from the proof of an analogous theorem of Geelen et~al.~\cite{GeelenGW02} that applies to branch-decompositions, and thus, e.g., to parameters known as \emph{branch-width} and \emph{carving-width}.

The proof of the upper bound on the number of types essentially boils down to the following setting. We are given a graph $G$ and a subset $X$ of vertices, such that at most $\ell$ edges have exactly one endpoint in $X$.
The question is how $X$ can look like in an optimum-width ordering of $G$. We prove that there is always an ordering where $X$ is split into at most $\Oh(k\ell)$ blocks, where $k$ is the optimum width. 
This allows us to store the relevant information on the whole $X$ in one of a constant number of types (called {\em{bucket interfaces}}).
The swapping argument used in this proof holds the essence of the typical sequences technique of Bodlaender and Kloks~\cite{BodlaenderK96}, while being, in our opinion, more natural and easier to understand.
%

As an interesting byproduct, we can also use our understanding to treat the problem of removing edges to get a graph of small cutwidth.
More precisely, for parameters $w,k$, we consider the class of all graphs $G$, such that $w$ edges can be removed from $G$ to obtain a graph of cutwidth at most $k$.
We prove that for every constant $k$, the minimal (strong) immersion obstructions for this class have sizes bounded {\em{linearly}} in $w$.  Moreover we  give an exponential lower bound to the number of these obstructions. These results are presented in Section~\ref{sec:remddist}.

\subsection{Algorithmic results.} Consider the following ``compression'' problem: given a graph $G$ and its ordering $\sigma$ of width $\ell$, we would like to construct, if possible, a new ordering 
of the vertices of $G$ of width at most $k$, where $k<\ell$.
Then the types defined above essentially match states that would be associated with prefixes of $\sigma$ in a dynamic programming algorithm solving this problem. Alternatively, one can think 
of building an automaton that traverses the ordering $\sigma$ of width $\ell$ while constructing an ordering of $G$ of width at most $k$.
Hence, our upper bound on the number of types can be directly used to limit the state space in such a dynamic programming procedure/automaton, yielding an FPT algorithm for the above problem.

With this result in hand, it is not hard to design of an exact FPT algorithm for cutwidth. 
One could introduce vertices one by one to the graph, while maintaining an ordering of optimum width.
Each time a new vertex is introduced, we put it anywhere into the ordering, and it can be argued that the new ordering has width at most three times larger than the optimum.
Then, the dynamic programming algorithm sketched above can be used to ``compress'' this approximate ordering to an optimum one in linear FPT time.

The above approach yields a quadratic algorithm. To match the optimum, linear running time, 
 we use a similar trick as Bodlaender in his
linear-time algorithm for computing the treewidth of the graph~\cite{Bodlaender96}. Namely, we show that instead of processing vertices one by one, 
we can proceed recursively by removing a significant fraction of all the edges at each step, so that their reintroduction increases the width at most twice. 
We then run the compression algorithm on the obtained 2-approximate ordering to get an optimum one.
The main point is that, since we remove a large portion of the graph at each step, the recursive equation on the running time solves to a linear function, instead of quadratic. 
This gives the following.

\begin{theorem}\label{thm:algo}
There exists an algorithm that, given an $n$-vertex graph $G$ and an integer $k$, runs in time $2^{\Oh(k^2\log k)}\cdot n$ and either correctly concludes that the cutwidth of $G$ is larger than $k$, 
or outputs an ordering of $G$ of width at most $k$.
\end{theorem}

The algorithm of Theorem~\ref{thm:algo} has running time slightly larger than that of Thilikos, Bodlaender, and Serna~\cite{ThilikosSB05,ThilikosSB05a}. 
The difference is the $\log k$ factor in the exponent, the reason for which is that we use a simpler bucketing approach to bound the number of states, instead of the more entangled, 
but finer, machinery of typical sequences. 
We believe the main strength of our approach lies in its explanatory character. 
Instead of relying on algorithms for computing tree or path decompositions, which are already difficult by themselves, and then designing a dynamic programming algorithm on a path decomposition,
we directly approach cutwidth ``via cutwidth'', and not ``via pathwidth''. That is, the dynamic programming procedure for computing the optimum cutwidth ordering on an approximate cutwidth ordering 
is technically far simpler and conceptually more insightful than performing the same on a general path decomposition.
We also show that the ``reduction-by-a-large-fraction'' trick of Bodlaender~\cite{Bodlaender96} can be performed also in the cutwidth setting, yielding a self-contained, natural, and understandable algorithm.

\section{Preliminaries}\label{sec:prelims}

We denote the set of non-negative integers by $\NN$ and the set of positive integers by $\NN^+$.
For $r,s\in \NN$ with $r\leq s$, we denote $[r]=\{1,\ldots,r\}$ and $\intv{r}{s}=\{r,\ldots,s\}$.
Notice that $[0]=\emptyset$.

\paragraph{Graphs.}
All graphs considered in this paper are undirected, without loops, and may have multiple edges.
The vertex and edge sets of a graph $G$ are denoted by $V(G)$ and $E(G)$, respectively.
For disjoint $X,Y\subseteq V(G)$, by $E_G(X,Y)$ we denote the set of edges of $G$ with one endpoint in $X$ and one in $Y$.
If $S\subseteq V(G)$, then we denote $\delta_{G}(S)=|E_G(S,V(G)\setminus S)|$. 
We drop the subscript if it is clear from the context.
Every partition $(A,B)$ of $V(G)$ is called a {\em{cut of $G$}}; the {\em{size}} of the cut $(A,B)$ is $\delta(A)$.

%

\paragraph{Cutwidth.}
Let $G$ be a graph and $\sigma$ be an ordering of $V(G)$.
For $u,v\in V(G)$, we write $u<_{\sigma} v$ if $u$ appears before $v$ in $\sigma$.
Given two disjoint sequences $\sigma_{1}=\langle x_1,\ldots,x_{r_{1}}\rangle$ 
and $\sigma_{2}=\langle y_1,\ldots,y_{r_{2}}\rangle$ of vertices in $V(G)$, we define their {\em{concatenation}} as $\sigma_1\circ \sigma_{2}=\langle x_1,\ldots,x_{r_{1}},y_1,\ldots,y_{r_{2}}\rangle$.
For $X\subseteq V(G)$, let $\sigma_X$ be the ordering of $X$ induced by $\sigma$, i.e., the 
ordering obtained from $\sigma$ if we remove the vertices that do not belong in $X$.
For a vertex $v$ we denote by $V_{v}^{\sigma}$ the set $\{u\in V(G) \mid u \leq_{\sigma} v\}$.
A \emph{$\sigma$-cut} is any cut of the form $(V^\sigma_v,V(G) \setminus V^\sigma_v)$ for $v\in V(G)$.
The {\em cutwidth of an ordering $\sigma$ of $G$} is defined as $\ctw_{\sigma}(G) = \max_{v \in V(G)} \delta(V^{\sigma}_{v})$. The {\em cutwidth of $G$}, $\ctw(G)$, is the minimum of $\ctw_{\sigma}(G)$ over
all possible orderings of $V(G)$. 


\paragraph{Obstructions.}
Let $\leq $ be a partial order on graphs. 
We say that $G'\lneqq G$ if $G'\leq G$ and $G'$ is not isomorphic to $G$.
A graph class ${\cal G}$ is {\em{closed under $\leq$}} if whenever $G'\leq G$ and $G\in {\cal G}$, we also have that $G'\in {\cal G}$.
Given a partial order $\leq$
and a graph class ${\cal G}$ closed under $\leq$, we define the {\em{(minimal) obstruction
set}} of ${\cal G}$ w.r.t. $\leq $, denoted by ${\bf obs}_{\leq}({\cal G})$, as the set containing all graphs 
where the following two conditions hold: 

\begin{itemize}
\item[]\hspace{-.97cm} {\bf O1}: $G\not\in {\cal G}$, i.e., $G$ is not a member of ${\cal G}$, and 

\item[]\hspace{-.97cm} {\bf O2}: for each $G'$ with $G'\lneqq G$, we have that $G'\in{\cal G}$.
\end{itemize}

 We say that a set of graphs ${\cal H}$ is a {\em  $\leq$-antichain} if  it does not contain any  pair of comparable elements wrt. $\leq$. By definition, for any class 
 ${\cal G}$ closed under $\leq$, the set ${\bf obs}_{\leq}({\cal G})$ is an antichain. 
\looseness=-1

\paragraph{Immersions.}
Let $H$ and $G$ be graphs. We say that $G$ contains $H$ as an
\emph{immersion} if there is a pair of functions $(\phi, \psi)$, called
an $H$-\emph{immersion model of $G$}, such that $\phi$ is an injection from $V(H)$ to $V(G)$ and
$\psi$ maps every edge $uv$ of $H$ to a path of $G$ between
$\phi(u)$ and $\phi(v)$ so that
different edges are mapped to edge-disjoint paths.
Every vertex in the image of $\phi$ is called a \emph{branch vertex}. 
If we additionally demand that no internal vertex of a path in $\psi(E(H))$ is a branch vertex, then we say that  $(\phi, \psi)$
is a {\em strong $H$-immersion model} and $H$ is a 
{\em strong immersion} of $G$. We denote by $H\leq_{\rm i} G$ ($H\leq_{\rm si} G$) the fact that $H$ is an immersion (strong immersion) of $G$; these are partial orders.
Clearly, for any two graphs $H$ and $G$, if $H\leq_{\rm si}G$ then $H\leq_{\rm i}G$. This  implies the following observation:

\begin{observation}
\label{osesimme}
If ${\cal G}$ is a graph class closed under $\leq_{\rm i}$, 
then  ${\bf obs}_{\leq_{\rm i}}({\cal G})\subseteq  {\bf obs}_{\leq_{\rm si}}({\cal G})$.
\end{observation}

Robertson and Seymour proved in~\cite{RobertsonS10} that every 
$\leq_{\rm i}$-antichain is finite and conjectured the same for $\leq_{\rm si}$.
It is well-known that for every $k\in \mathbb{N}$, the class ${\cal C}_{k}$ of graphs of cutwidth at most $k$ is closed under immersions.
It follows from the results of~\cite{RobertsonS10} that  $\obs_{\leq_{\rm i}}(\C_{k})$ is finite; the goal of this paper is to provide good estimates on the sizes of graphs in $\obs_{\leq_{\rm si}}({\cal C}_{k})$.
As the cutwidth of a graphs is the maximum cutwidth of its connected components, it follows that graphs in $\obs_{\leq_{\rm si}}(\C_{k})$ are connected.
Moreover, every graph in $\obs_{\leq_{\rm si}}(\C_{k})$ has cutwidth exactly $k+1$, because the removal of any of its edges decreases its cutwidth to at most $k$.

\section{Bucket interfaces}\label{sec:interfaces}

Let $G$ be a graph and $\sigma$ be an ordering of $V(G)$. 
For a set $X\subseteq V(G)$, the {\em{$X$-blocks}} in $\sigma$ are the maximal subsequences of consecutive vertices of $\sigma$ that belong to $X$.
Suppose $(A,B)$ is a cut of $G$. Then we can write
 $\sigma = b_1\circ\ldots\circ b_p,$
where $b_1,\ldots,b_p$ are the $A$- and $B$-blocks in $\sigma$; these will be called jointly {\em{$(A,B)$-blocks}}.
%
The next lemma is the cornerstone of our approach: we prove that given a graph $G$ and a cut $(A,B)$ of $G$, 
there exists an optimum cutwidth ordering of $G$ where number of blocks depends only on the cutwidth and the size of $(A,B)$.


\begin{lemma}\label{lem:num_blocks}
Let $\ell\in \mathbb{N}^+$ and $G$ be a graph. If $(A,B)$ is a cut of $G$ of size $\ell$, 
then there is an optimum cutwidth ordering $\sigma$ of $V(G)$ with 
at most $(2\ell+1) \cdot (2\ctw(G)+3)+2\ell$ $(A,B)$-blocks.
\end{lemma}

\begin{proof}
	Let $\sigma$ be an optimum cutwidth ordering such that, subject to the width being minimum, the number of $(A,B)$-blocks it defines is also minimized.
	Let $\sigma = b_1 \circ b_2 \circ \dots \circ b_r$, where $b_1, b_2, \dots, b_r$ are the $(A,B)$-blocks of $\sigma$.
	If $\sigma$ defines less than three blocks, then the claim already follows, so let us assume $r \geq 3$.
	
	Consider any ordering $\sigma'$ obtained by swapping two blocks, i.e.,  $\sigma' = b_1 \circ \dots \circ b_{j-1} \circ b_{j+1} \circ b_j \circ b_{j+2} \dots b_r$, for some $j \in [r-1]$.
	Observe that since the blocks $b_1, \dots, b_r$ alternate as $A$-blocks and $B$-blocks, the ordering $\sigma'$ 	has a strictly smaller number of blocks;
	indeed, either $j-1 \geq 1$, in which case $b_{j-1} \circ b_{j+1}$ defines a single block of $\sigma'$, or $j=1$ and hence $j+2 \leq r$, in which case $b_j \circ b_{j+2}$ does.
	Therefore, by choice of $\sigma$, for each $j\in [r-1]$, swapping $b_j$ and $b_{j+1}$ in $\sigma$ must yield an ordering with strictly larger cutwidth.
	
	We call a block \emph{free} if it does not contain any endpoint of the cut edges $E_G(A,B)$.
We  now prove that any sequence of consecutive free blocks in $\sigma$ has at most $2\ctw(G)+3$ blocks.
Since the cut $(A,B)$ has size $\ell$, there are at most $2\ell$ blocks that are not free. 
This implies the claimed bound on the total number of all blocks in $\sigma$.

Suppose, to the contrary, that there exists a sequence of $q>2\ctw(G)+3$ consecutive free blocks in $\sigma$. Let these blocks be
$b_r,b_{r+1},\ldots,b_{s}$, where $s-r+1=q$.
For $j \in [r,s-1]$, we define $\mu(j)$ to be the size of the cut between all vertices inside or preceding the vertices of block $b_j$ and all vertices inside or following the vertices
of block $b_{j+1}$ in $\sigma$; see Figure~\ref{fig:num-blocks}.

\begin{claim}\label{clm:clm1}
For all $j\in [r+1,\dots,s-2]$, we have that $\mu(j-1) > \mu(j)$ or $\mu(j) < \mu(j+1)$. 
\end{claim}
\begin{proof}
Suppose that for some $j\in [r+1,s-2]$, 
$\mu(j) \geq \max(\mu(j-1),\mu(j+1))$. 
We will then show that the ordering $\sigma'$ obtained by swapping the blocks $b_{j}$ and $b_{j+1}$ still has optimum cutwidth, a contradiction to the choice of $\sigma$.
Notice that for every vertex $v$ preceding all vertices of $b_{j}$ or succeeding all vertices of $b_{j+1}$, $\delta(V^{\sigma'}_{v})=\delta(V^{\sigma}_{v})$.
Thus, it remains to show that for any vertex $v$ belonging to the block $b_{j}$ or to the block $b_{j+1}$, also $\delta(V^{\sigma'}_{v})\leq \delta(V^{\sigma}_{v})$. 

Let $p_{j}$  be the number of edges of $G$ with one endpoint in the block $b_{j}$ and the other endpoint preceding (in $\sigma$) all vertices of $b_{j}$.
Let also $s_{j}$ be the number of edges of $G$ with one endpoint in $b_j$ and the other endpoint succeeding (in $\sigma$) all vertices of $b_{j}$ (and hence succeeding all vertices of block $b_{j+1}$, since both $b_{j}$ and $b_{j+1}$ are free).
Notice that $\mu(j)=\mu(j-1) - p_{j} + s_{j}$ and recall that $\mu(j)\geq\mu(j-1)$. This yields that $s_{j} \geq p_{j}.$

Similarly, let $p_{j+1}$  be the number of edges of $G$ with one endpoint in $b_{j+1}$ and the other endpoint preceding all vertices of the block $b_{j+1}$ 
(and, in particular, all vertices of block $b_{j}$).
Let also $s_{j+1}$ be the number of edges of $G$ with one endpoint in $b_{j+1}$ and the other endpoint succeeding all vertices of block $b_{j+1}$. Again, we have $\mu(j+1)=\mu(j) - p_{j+1} + s_{j+1}$ and $\mu(j)\geq\mu(j+1)$. This yields that $p_{j+1} \geq s_{j+1}.$

Let $v$ be a vertex of the block $b_{j}$.
Recall that the blocks $b_{j}$ and $b_{j+1}$ are free and thus, there are no edges between them.
Observe then that $\delta(V^{\sigma'}_{v})=\delta(V^{\sigma}_{v})+s_{j+1}-p_{j+1}\leq \delta(V^{\sigma}_{v})$.
Symmetrically, for any vertex $v$ in $b_{j+1}$, observe that $\delta(V^{\sigma'}_{v})=\delta(V^{\sigma}_{v})+p_{j}-s_{j}\leq \delta(V^{\sigma}_{v})$.
Thus, $\ctw_{\sigma'}(G)\leq \ctw_{\sigma}(G)=\ctw(G)$, a contradiction.
\cqed\end{proof}
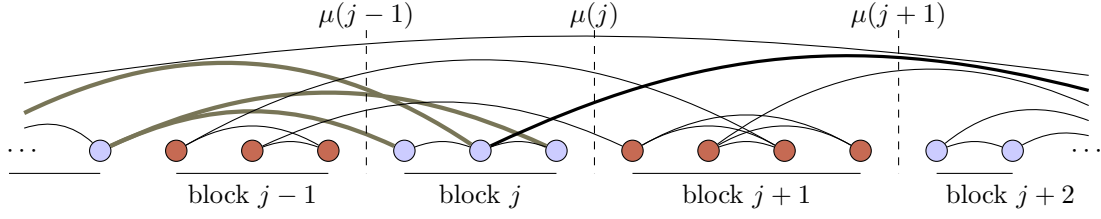
\begin{figure}[t]
	\centering

\begin{tikzpicture}[scale = 1]
	\node at (0,0) {$\cdots$};
	\node[A] (v1) at (1,0) {};
	\draw (-0.2,-0.3) -- (1,-0.3);
	\node[B] (v2) at (2,0) {};
	\node[B] (v3) at (3,0) {};
	\node[B] (v4) at (4,0) {};
	\draw (2,-0.3) -- (4,-0.3);
	\node[A] (v5) at (5,0) {};
	\node[A] (v6) at (6,0) {};
	\node[A] (v7) at (7,0) {};
	\draw (5,-0.3) -- (7,-0.3);
	\node[B] (v8) at (8,0) {};
	\node[B] (v9) at (9,0) {};
	\node[B] (v10) at (10,0) {};
	\node[B] (v11) at (11,0) {};
	\draw (8,-0.3) -- (11,-0.3);
	\node[A] (v12) at (12,0) {};
	\node[A] (v13) at (13,0) {};
	\draw (12,-0.3) -- (13,-0.3);
	\node at (14,0) {$\cdots$};

	\draw (0,0.9) to [bend left=8] (14,1.0);

	\draw (0,0.3) to[bend left] (v1);
	\draw (v2) to[bend left] (v4);
	\draw (v3) to[bend left=20] (v4);

	\draw[yellow!35!black,ultra thick] (0,0.5) to [bend left] (v6);
	\draw[yellow!35!black,ultra thick] (v1) to[bend left=25] (v5);
	\draw[yellow!35!black,ultra thick] (v1) to[bend left=25] (v7);
	\draw (v3) to[bend left=25] (v8);

	\draw (v5) to [bend left=20] (v6);
	\draw (v6) to [bend left=20] (v7);

	\draw[very thick] (v6) to [bend left=20] (14, 0.8);
	\draw (v2) to [bend left] (v10);

	\draw (v8) to [bend left] (v10);
	\draw (v9) to [bend left=20] (v10);
	\draw (v9) to [bend left] (v11);
	\draw (v8) to [bend left] (v11);
	\draw (v9) to [bend left] (14,0.6);

	\draw (v13) to [bend left=20] (14,0.2);

	\draw (v12) to [bend left=20] (v13);
	\draw (v12) to [bend left] (14,0.4);

	\node at (3,-0.6) {block $j-1$};
	\node at (4.5,1.8) {$\mu(j-1)$};
	\draw[dashed] (4.5,1.6)--(4.5,-0.3);
	\node at (6,-0.6) {block $j$};
	\node at (7.5,1.8) {$\mu(j)$};
	\draw[dashed] (7.5,1.6)--(7.5,-0.3);
	\node at (9.5,-0.6) {block $j+1$};
	\node at (11.5,1.8) {$\mu(j+1)$};
	\draw[dashed] (11.5,1.6)--(11.5,-0.3);
	\node at (12.5,-0.6) {~~~~~~~~block $j+2$};
\end{tikzpicture}

	\caption{A cut $(A,B)$ is highlighted (blue, red), with the corresponding blocks underlined and cuts between them marked with dashed lines. 
	Edges counted as $p_{j}$ and $s_{j}$ are thickened.}
\label{fig:num-blocks}
\end{figure}

Claim~\ref{clm:clm1} shows that for all $j\in [r+1,s-2]$, we have $\mu(j-1) > \mu(j)$ or $\mu(j) < \mu(j+1)$. It follows that any non-decreasing pair $\mu(j-1)\leq \mu(j)$ 
must be followed by an increasing pair $\mu(j) < \mu(j+1)$. Hence, if $j_{\min}$ is the minimum index such that $\mu(j_{\min})\leq \mu(j_{\min}+1)$, 
then the sequence $\mu(j)$ has to be strictly decreasing up to $j_{\min}$ and strictly increasing from $j_{\min}+1$ onward. Since $\mu(j) \leq \ctw(G)$ for all $j$, 
the length $q$ of the sequence of consecutive free blocks cannot be longer than $2\ctw(G)+3$ in total, concluding the proof.
\end{proof}

We use the above lemma to bound the number of ``types'' of prefixes in graph orderings.
To describe such a prefix, i.e., one side of a cut in a graph, we use the following definition.

\begin{definition}\label{def:boundaried-graph}
	A \emph{$k$-boundaried graph} is a pair $\mathbf{G}=(G,\bar{x})$ where $G$ is a graph and $\bar{x}=(x_{1},\dots,x_k)$ is a $k$-tuple of the graph's {\em{boundary vertices}} (ordered, not necessarily distinct).
	The \emph{extension} of $\mathbf{G}$ is the graph $G^{*}$ obtained from $G$ by adding $k$ new vertices $x_{1}',\dots,x_k'$ and edges $x_{1} x_{1}', \dots, x_k x_k'$.
	The \emph{join} $\mathbf{A} \oplus \mathbf{B}$ of two $k$-boundaried graphs  $\mathbf{A} = (A, \bar{x}), \mathbf{B}=(B,\bar{y})$ is the graph obtained from the disjoint union of $A$ and $B$ by adding an edge $x_{i} y_{i}$ for  $i\in[k]$.
\end{definition}

From Lemma~\ref{lem:num_blocks} we derive that for any given cut $(A,B)$ of size $\ell$ of a graph $G$ with $\ctw(G)\leq k$, there is an optimum cutwidth ordering in which 
the vertices of $A$ occur in $\Oh(k\ell)$ blocks. 
Our next goal is to show that the only information about $A$ that can affect the cutwidth of $G$ is: the placing of the endpoints of each cutedge ($x_{i}$ and $x_{i}'$) into blocks, and 
the cutwidth of each block (as an induced subgraph of $A$ or $A^{*}$).
%
Recall that for an ordering $\sigma$ of $V(G)$, \emph{$\sigma$-cuts} are cuts of the form $(V^\sigma_v, V(G)\setminus V^\sigma_v)$, for $v \in V(G)$.

\begin{definition}
	Let $G$ be a graph and $\sigma$ be an ordering of its vertices. An \emph{$\ell$-bucketing} of $\sigma$ is a function 
	$T\colon V(G) \to [\ell]$ such that $T(u)\leq T(v)$ for any $u$ appearing before $v$ in $\sigma$. 
	For every $i\in [\ell]$, the set $T^{-1}(i)$ will be called a {\em bucket}; a bucket is naturally ordered by $\sigma$. 
	For every bucket $T^{-1}(i)$, $i \in [\ell]$, let $\cuts(G,\sigma,T,i)$ be the family of $\sigma$-cuts containing on one side all vertices of buckets appearing before 
	$i$ and a prefix (in $\sigma$) of the $i$-th bucket.
	For an ordering $\sigma$ of the vertices of a graph $G$, define the \emph{width} of the bucket $i$, $i\in[\ell]$, as the maximum width of any cut in the family $\cuts(G,\sigma,T,i)$.
	Formally,
	\begin{eqnarray*}
	\cuts(G,\sigma,T,i) & = & \left\{ \left(T^{-1}([1,i-1])\ \cup\ L,\quad R\ \cup\ T^{-1}([i+1,\ell])\right) \colon \right. \\
	& &~~\left. (L,R)\mbox{ is a }\sigma\mbox{-cut of }T^{-1}(i)\right\},\\
	\width(G,\sigma,T,i) & = & \max \left\{\, \left|E_G(L,R)\right|\ \colon\ (L,R) \in \cuts(G,\sigma,T,i)\, \right\}.
	\end{eqnarray*}
\end{definition}

\noindent
Notice that every $\sigma$-cut of $G$ is in $\cuts(G,\sigma,T,i)$ for at least one bucket $i \in [\ell]$;
since $\ctw_\sigma(G)$ is the maximum of $\left|E_G(L,R)\right|$ over $\sigma$-cuts $(L,R)$, we have
\begin{equation}\label{eq:cwdthsgm}\ctw_\sigma(G) = \max_{i\in[\ell]}\ \width(G,\sigma,T,i).\end{equation}
For two $k$-boundaried graphs $\mathbf{A}=(A,\bar{x}),\mathbf{B}=(B,\bar{y})$, we slightly abuse notation and understand the edges $x_{1}x_{1}',\dots,x_kx_k'$ in $A^{*}$ to be the same as $y_{1}'y_{1},\dots,y_k'y_k$ in $B^{*}$ and as $x_{1} y_{1},\dots, x_k y_k$ in $\mathbf{A}\oplus\mathbf{B}$.
That is, for an ordering $\sigma$ of $\mathbf{A}\oplus\mathbf{B}$ with $\ell$-bucketing $T$,
we define $T|_{A^{*}}(v)$ to be $T(v)$ for $v \in V(A)$ and $T(y_{i})$ for $v=x_{i}'$.
We define $\sigma|_{A^{*}}$ as an ordering that orders $x_{i}'$ just as $\sigma$ orders $y_{i}$, with the order between $x_{i}'$ and $x_j'$ chosen arbitrarily when $y_{i}=y_j$.
The following lemma shows that if an $\ell$-bucketing respects the sides of a cut, then the width of any bucket can be computed as the sum of contributions of the sides.

%

\begin{lemma}\label{lem:width_by_parts}
	Let $k,\ell$ be positive integers and $\mathbf{A}=(A,\bar{x}),\mathbf{B}=(B,\bar{y})$ be two $k$-boundaried graphs.
	Let also $\sigma$ be a vertex ordering of $\mathbf{A}\oplus\mathbf{B}$ with $\ell$-bucketing $T$.
	If $T^{-1}(i)$ does not contain any vertex of $A$, for some $i\in [\ell]$,
	that is, $T^{-1}(i) \cap V(A) = \emptyset$, 
	then it holds that
	$\width(\mathbf{A}\oplus \mathbf{B}, \sigma, T, i) = 
	\width(A, \sigma|_{A}, T|_{A}, i) + \width(B^{*}, \sigma|_{B^{*}}, T|_{B^{*}}, i)$.
\end{lemma}

\begin{proof}
	Consider any cut $(L,R)$ in $\cuts(G,\sigma,T,i)$.
	Observe that for every edge $e$ of $E_{\mathbf{A}\oplus \mathbf{B}}(L,R)$ one of the following holds:
	\begin{enumerate}
	\item $e\in E_{A}(L \cap V(A), R\cap V(A))$ or
	\item $e\in E_{B}(L \cap V(B), R \cap V(B))$ or
	\item $e\in E_{G}(L\cap V(A), R\cap V(B))$, or
	\item $e\in E_{G}(R\cap V(A), L\cap V(B))$.
	\end{enumerate}
	Since we do not distinguish between the vertices $x_{i}$ and the vertices $y_{i}'$, we equivalently obtain that for every edge $e\in E_{\mathbf{A}\oplus \mathbf{B}}(L,R)$,
	$e$ is either an edge in $E_A(L \cap V(A), R\cap V(A))$ or an edge in $E_{B^{*}}(L \cap V(B^{*}), R \cap V(B^{*}))$.	
	Observe that $(L \cap V(A), R\cap V(A))$ is a cut in $\cuts(A,\sigma|_{A},T|_{A},i)$ and 
	$(L \cap V(B^{*}), R\cap V(B^{*}))$ is a cut in $\cuts(B^{*},\sigma|_{B^{*}},T|_{B^{*}},i)$.
	Therefore, the total number of edges crossing these cuts is at most $\width(A, \sigma|_{A}, T|_{A}, i) + \width(B^{*}, \sigma|_{B^{*}}, T|_{B^{*}}, i)$.
	This proves that $$\width(\mathbf{A}\oplus \mathbf{B}, \sigma, T, i) \leq 
	\width(A, \sigma|_{A}, T|_{A}, i) + \width(B^{*}, \sigma|_{B^{*}}, T|_{B^{*}}, i).$$
	
	For the converse inequality, observe that since the bucket $T^{-1}(i)$ does not contain any vertices of $A$, we have $T|_{A}^{-1}(i) = \emptyset$.
	Hence there is exactly one cut in $\cuts(A,\sigma|_A,T|_A,i)$, namely $(L_A,R_A)$, where $L_A = T^{-1}(\{1,\dots,i-1\})\cap V(A)$ and $R_A=T^{-1}(\{i+1,\dots,\ell\}) \cap V(A)$.
	Let $(L_B,R_B)$ be a cut in $\cuts(B^{*},\sigma|_{B^{*}},T|_{B^{*}},i)$ maximizing $|E_{B^{*}}(L_B,R_B)|$.
	Then, since we assumed that $T^{-1}(i)$ does not contain any vertices of $A$ (and thus, may only contain vertices of $B$), it follows that
	$(L_A \cup L_B, R_A \cup R_B)$ is a cut in $\cuts(G,\sigma,T,i)$.
	As above, every edge of $\mathbf{A}\oplus\mathbf{B}$ crossing this cut is either in $E_A(L_A,R_A)$ or in $E_{B^{*}}(L_B,R_B)$, where we again do not distinguish between 
	the vertices $x_{i}$ and $y_{i}'$.
	Hence 
\begin{eqnarray*}
		\width(\mathbf{A}\oplus \mathbf{B}, \sigma, T, i) &  \geq & |E_{\mathbf{A}\oplus \mathbf{B}}(L,R)| \\ 
		& =&  |E_A(L_A,R_A)| + |E_{B^{*}}(L_B,R_B)|\\
		&  = & \width(A,\sigma|_A,T|_A,i) + \width(B^{*}, \sigma|_{B^{*}}, T|_{B^{*}}, i).
\end{eqnarray*}
\end{proof}

Replacing the roles of $\mathbf{A}$ and $\mathbf{B}$ above, we obtain that if $T^{-1}(i)$ does not contain any vertex of $B$, 
then $$\width(\mathbf{A}\oplus \mathbf{B}, \sigma, T, i) = \width(A^{*}, \sigma|_{A^{*}}, T|_{A^{*}}, i) + \width(B, \sigma|_{B}, T|_{B}, i).$$
Intuitively, this implies that the cutwidth of $\mathbf{A}\oplus\mathbf{B}$ depends on $A$ only in the widths of each block relative to $A$ and $A^{*}$ (in any bucketing where
buckets are either $A$-blocks or $B$-blocks).
Therefore, replacing $\mathbf{A}$ with another boundaried graph whose extension has an ordering and bucketing with the same widths preserves cutwidth (as long as endpoints of the cut edges are placed in the same buckets too).
This is formalized in the next definition.

\begin{definition}\label{def:interface}
	For $k,\ell\in\mathbb{N}$, a \emph{($k$,$\ell$)-bucket interface} consists of functions:
	\begin{itemize}
		\item $b,b': [k] \to [\ell]$ identifying the buckets which contain $x_{i}$ and $x_{i}'$, respectively and
		\item $\mu,\mu^{*}: [\ell] \to [0,k]$ corresponding to the widths of buckets. 
	\end{itemize}
	A $k$-boundaried graph $\mathbf{G}$ \emph{conforms} with a $(k,\ell)$-bucket interface
	if there exists an ordering $\sigma$ of the vertices of $G^{*}$ and an $\ell$-bucketing $T$ of $G^*$ such that:
	\begin{itemize}
		\item $T(v)$ is odd for $v\in V(G)$ and even for $v\in \{x_{1}',\dots,x_k'\}$,
		\item $T(x_{i}) = b(i)$ and $T(x_{i}') = b'(i)$,\ for each $i\in [k]$,
		\item $\width(G, \sigma|_{G}, T|_{G}, j) \leq \mu(j)$,\ for each $j\in [\ell]$,
		\item $\width(G^{*}, \sigma, T, j) \leq \mu^{*}(j)$,\ for each $j\in [\ell]$.
	\end{itemize}
\end{definition}

\begin{observation}\label{obs:bcketsze} For all $k,\ell\in\NN^+$ there are $
\leq 2^{2(k\log\ell +\ell\log(k+1))}$ $(k,\ell)$-bucket interfaces.
\end{observation}

We call two $k$-boundaried graphs $\mathbf{G}_{1}, \mathbf{G}_{2}$ \emph{($k$,$\ell$)-similar} if the sets of $(k,\ell)$-bucket interfaces they conform with are equal.
The following lemma subsumes the above ideas. The proof follows easily from Lemma~\ref{lem:width_by_parts} and the fact that $\ctw_\sigma(G) = \max_{i\in[\ell]}\ \width(G,\sigma,T,i)$ (Eq.~\eqref{eq:cwdthsgm}).

\begin{theorem}\label{thm:myhill_nerode}
	Let $k,r$ be two positive integers. Let also $\mathbf{A}_{1}$ and $\mathbf{A}_{2}$ be two $k$-boundaried graphs that are $(k,\ell)$-similar,
	where $\ell=(2k+1) \cdot (2r+4)$.
	Then for any $k$-boundaried graph $\mathbf{B}$ where $\ctw(\mathbf{A}_{1}\oplus \mathbf{B})\leq r$, it holds that 
	$\ctw(\mathbf{A}_{2}\oplus \mathbf{B})=\ctw(\mathbf{A}_{1}\oplus \mathbf{B})$.
\end{theorem}

\begin{proof}
	Let $\mathbf{A}_{i}=(A_{i},\bar{x}^i), \mathbf{B}=(B,\bar{y})$ and suppose that $\ctw(\mathbf{A}_{1}\oplus \mathbf{B})\leq r$.
	By Lemma~\ref{lem:num_blocks}, there is an optimum cutwidth ordering $\sigma_{1}$ of the vertices of $\mathbf{A}_{1}\oplus \mathbf{B}$ that has at most $\ell-1$ $(V(A_{1}),V(B))$-blocks. 
	In particular $\ctw_{\sigma_{1}}(\mathbf{A}_{1}\oplus \mathbf{B})=\ctw(\mathbf{A}_{1}\oplus \mathbf{B})\leq r$.
	By adding an empty block at the front, if necessary, we may assume that the number of blocks is at most $\ell$, while odd-indexed blocks are $V(A_1)$-blocks and even-indexed blocks are $V(B)$-blocks. 
	Then, there is an $\ell$-bucketing $T_{1}$ of $\sigma_{1}$ such that
	$T_{1}(v)$ is odd for $v \in A_{1}$ and even for $v \in B$.
	Therefore $\sigma_{1}|_{A_{1}^{*}}$ and $T_{1}|_{A_{1}^{*}}$ certify that the following $(k,\ell)$-bucket interface conforms with $\mathbf{A}_{1}$:
	\begin{itemize}
		\item $b(i) = T_{1}(x^{1}_{i})$ and $b'(i)=T_{1}|_{A_{1}^{*}}({x^{1}_{i}}')=T_{1}(y_{i})$ for $i\in[k]$,
		\item $\mu(i) = \width(A_{1}, \sigma_{1}|_{A_{1}}, T_{1}|_{A_{1}}, i)$ and $\mu^{*}(i) = \width(A_{1}^{*}, \sigma_{1}|_{A_{1}^{*}}, T_{1}|_{A_{1}^{*}}, i)$ for $i \in [\ell]$.
	\end{itemize}
	
	By $(k,\ell)$-similarity there is an ordering $\sigma_{2}$ of $A_{2}^{*}$ and its $\ell$-bucketing $T_{2}$ such that:
	\begin{itemize}
		\item each bucket $T_{2}^{-1}(i)$ is contained in $A_{2}$ for odd $i\in[\ell]$ and in $\{{x^{2}_{1}}',\dots,{x^{2}_{k}}'\}$ for even $i\in[\ell]$
		\item $b(i) = T_{2}(x^{2}_{i})$ and $b'(i)=T_{2}({x^{2}_{i}}')$ for $i\in[k]$,
		\item $\mu(i) \geq \width(A_{2}, \sigma_{2}|_{A_{2}}, T_{2}|_{A_{2}}, i)$ and $\mu^{*}(i) \geq \width(A_{2}^{*}, \sigma_{2}|_{A_{2}^{*}}, T_{2}|_{A_{2}^{*}}, i)$ for $i \in [\ell]$.
	\end{itemize}
	Given this, we define an assignment of vertices into buckets $\Pi\colon V(\mathbf{A}_{2} \oplus \mathbf{B}) \to [\ell]$ as follows.
	\begin{itemize}
		\item $\Pi(v) = T_{1}(v)$ for $v \in V(B)$ and 
		\item $\Pi(v) = T_{2}(v)$ for $v \in V(A_{2})$.
	\end{itemize}	
	Clearly,
	\begin{align}
	&\Pi|_B = T_{1}|_B \qquad \textrm{and} \label{eq:relakghraelgnawl}\\
	&\Pi|_{A_{2}} = T_{2}|_{A_{2}}.\label{eq:ejkrsghearkgh}
	\end{align}	
	We claim that $\Pi|_{A_{2}^{*}} = T_{2}|_{A_{2}^{*}}$ and $\Pi|_{B^{*}} = T_{1}|_{B^{*}}$ also hold. 
	Indeed, 
	\begin{align*}
	\Pi|_{A_{2}^{*}}(x^{2'}_{i}) &  = \Pi(y_{i}) & && && & && \mbox{(we consider } x^{2'}_{i} \mbox{ as } y_{i})\\
	                                    &=  T_{1}(y_{i}) & & & && & &&\mbox{(by definition)}\\
	                                   &  =  b'(i) & &&& && & & ((k,\ell)\mbox{-bucket interface)}\\
	                                  &   =  T_{2}(x^{2'}_{i}) && & && & && ((k,\ell)\mbox{-similarity)}
	\end{align*}
	
 \noindent	 and, similarly,
	  \begin{align*}
	  \Pi|_{B^{*}}(y_{i}') &  =  \Pi(x^{2}_{i}) & && && & &&\mbox{(we consider } y_{i}' \mbox{ as } x^{2}_{i})\\
	                            &  =  T_{2}(x^{2}_{i}) & && && & && \mbox{(by definition)}\\
	                            & =  b(i) & && && & && ((k,\ell)\mbox{-bucket interface)}\\
	                            & = T_{1}(x^{1}_{i}) & && && & && ((k,\ell)\mbox{-similarity)}\\
	                            & = T_{1}|_{B^{*}}(y_{i}') & && && & &&  \mbox{(by definition).}
	   \end{align*}
	   Thus, we obtain that
	   \begin{eqnarray}
	   \Pi|_{A_{2}^{*}} & = & T_{2}|_{A_{2}}\label{eq:eskrlgjhaekj}\\
	   \Pi|_{B^{*}} & = & T_{1}|_{B^{*}}.\label{eq:ergerglkjwag}
	   \end{eqnarray}

		Note also that vertices of $A_{2}$ are mapped to odd buckets and vertices of $B$ are mapped to even buckets.
We use $\Pi$ to define an ordering $\pi$ of the vertices of $\mathbf{A_{2}}\oplus \mathbf{B}$ as follows. Formally, we let $u<_{\pi} u$ if and only if one of the following conditions hold:
\begin{enumerate}
\item $\Pi(u) < \Pi(v)$,
\item $u <_{\sigma_{2}} v$ and $\Pi(u)=\Pi(v)$ is odd, or
\item $u <_{\sigma_{1}} v$ and $\Pi(u)=\Pi(v)$ is even.
\end{enumerate}

Note that this is a linear ordering as it first sorts the vertices according to the bucket they belong to and then according to the ordering induced in this bucket by the orderings $\sigma_{1}$
and $\sigma_{2}$. Note also that by definition $\Pi$ is an $\ell$-bucketing of $\pi$.
Recall that, from Eq.~\eqref{eq:eskrlgjhaekj}, $\Pi|_{A_{2}^{*}}=T_{2}|_{A_{2}}$. This, together with the observation that the vertices of $A_{2}$ are mapped to odd buckets of $\Pi$, implies that 
\begin{align}
& \pi|_{A_{2}^{*}}=\sigma_{2}|_{A_{2}^{*}}\qquad  \textrm{ and that}\label{eq:eqkdjgnsekjrga}\\
& \pi|_{A_{2}}=\sigma_{2}|_{A_{2}}.\label{eq:eqgaelrghera}
\end{align}
Moreover, recall that $\Pi|_{B^{*}} = T_{1}|_{B^{*}}$. This, together with the fact that the vertices of $B$ are mapped to even buckets of $\Pi$, implies that 
\begin{align}
& \pi|_{B^{*}} = \sigma_{1}|_{B^{*}} \qquad \textrm{and that}\label{eq:eqsegjkrserger}\\
& \pi|_{B} = \sigma_{1}|_{B}.\label{eq:eqeqwrjkawefl}
\end{align}
We now bound the width of each bucket. Let $i\in[\ell]$. Notice that if $i$ is even the by construction $\Pi^{-1}(i)$ contains only vertices from $B$.
Therefore, 
\begin{eqnarray}
	\width(\mathbf{A}_{2}\oplus\mathbf{B},\pi,\Pi,i) & = & \width(A_{2},\pi|_{A_{2}},\Pi|_{A_{2}},i) + \width(B^{*},\pi|_{B^{*}},\Pi|_{B^{*}},i)\nonumber \\
	& = & \width(A_{2},\sigma_{2}|_{A_{2}},T_{2}|_{A_{2}},i) + \width(B^{*},\sigma_{1}|_{B^{*}},T_{1}|_{B^{*}},i)\nonumber\\
	& \leq  & \mu(i) + \width(B^{*},\sigma_{1}|_{B^{*}},T_{1}|_{B^{*}},i) \nonumber\\
	& = & \width(A_{1},\sigma_{1}|_{A_{1}},T_{1}|_{A_{1}},i) + \width(B^{*},\sigma_{1}|_{B^{*}},T_{1}|_{B^{*}},i)\nonumber\\
	& = & \width(\mathbf{A}_{1}\oplus\mathbf{B},\sigma_{1},T_{1},i),\label{eq:wethweahwea}
\end{eqnarray}
where the first equality follows from Lemma~\ref{lem:width_by_parts}, the second equality holds by 
Eq.~\eqref{eq:ejkrsghearkgh},~\eqref{eq:eqgaelrghera},~\eqref{eq:eqsegjkrserger}, and~\eqref{eq:ergerglkjwag},
the third inequality follows from the $(k,\ell)$-bucket interface, and the fifth equality follows from Lemma~\ref{lem:width_by_parts}.
We similarly argue, using $\mu^{*}$ instead of $\mu$, that for odd $i\in[\ell]$, 
$\width(\mathbf{A}_{2}\oplus\mathbf{B},\pi,\Pi,i)=\width(\mathbf{A}_{1}\oplus\mathbf{B},\sigma_{1},T_{1},i).$
In particular,
\begin{eqnarray}
	\width(\mathbf{A}_{2}\oplus\mathbf{B},\pi,\Pi,i) & = & \width(A_{2}^{*},\pi|_{A_{2}^{*}},\Pi|_{A_{2}^{*}},i) + \width(B,\pi|_{B},\Pi|_{B},i)\nonumber \\
	& = & \width(A_{2}^{*},\sigma_{2}|_{A_{2}^{*}},T_{2}|_{A_{2}^{*}},i) + \width(B,\sigma_{1}|_{B},T_{1}|_{B},i)\nonumber\\
	& \leq  & \mu^{*}(i) + \width(B,\sigma_{1}|_{B},T_{1}|_{B},i) \nonumber\\
	& = & \width(A_{1}^{*}, \sigma_{1}|_{A_{1}^{*}}, T_{1}|_{A_{1}^{*}}, i) + \width(B,\sigma_{1}|_{B},T_{1}|_{B},i)\nonumber\\
	& = & \width(\mathbf{A}_{1}\oplus\mathbf{B},\sigma_{1},T_{1},i).\label{eq:eqdskjrgneskjrng}
\end{eqnarray}
Similarly, to Eq.~\ref{eq:wethweahwea}, we get that the first equality follows from Lemma~\ref{lem:width_by_parts}, the second equality holds by 
Eq.~\eqref{eq:eskrlgjhaekj},~\eqref{eq:eqkdjgnsekjrga},~\eqref{eq:relakghraelgnawl}, and~\eqref{eq:eqeqwrjkawefl},
the third inequality follows from the $(k,\ell)$-bucket interface, and the fifth equality follows from Lemma~\ref{lem:width_by_parts}.

Therefore, from Eq.~\eqref{eq:wethweahwea} and~\eqref{eq:eqdskjrgneskjrng} 
we obtain that
$$\ctw_\pi(\mathbf{A}_{2}\oplus\mathbf{B}) = \max_{i\in[\ell]}\  \width(\mathbf{A}_{2}\oplus\mathbf{B},\pi,\Pi,i) \leq \max_{i\in[\ell]}\  \width(\mathbf{A}_{1}\oplus\mathbf{B},\sigma_{1},T_{1},i) = \ctw_{\sigma_{1}}(\mathbf{A}_{1}\oplus\mathbf{B}).$$
Moreover, since $\ctw(\mathbf{A}_{2}\oplus\mathbf{B})\leq \ctw_\pi(\mathbf{A}_{2}\oplus\mathbf{B})$ and $\sigma_{1}$ is an optimum cutwidth ordering for $\mathbf{A}_{`}\oplus\mathbf{B}$,
it follows that $$\ctw(\mathbf{A}_{2}\oplus\mathbf{B})\leq\ctw(\mathbf{A}_{1}\oplus\mathbf{B})\leq r.$$
So in particular $\ctw(\mathbf{A}_{2}\oplus\mathbf{B})\leq r$. By applying the same reasoning, but with $\mathbf{A}_1$ and $\mathbf{A}_2$ reversed, we obtain also the converse inequality $\ctw(\mathbf{A}_{2}\oplus\mathbf{B})\leq\ctw(\mathbf{A}_{1}\oplus\mathbf{B})$.
This proves that indeed $\ctw(\mathbf{A}_{2}\oplus\mathbf{B})=\ctw(\mathbf{A}_{1}\oplus\mathbf{B})$.
\end{proof}

\section{Obstruction sizes and linked orderings}\label{sec:linked}

In this section we establish the main result on sizes of obstructions for cutwidth. 
We first introduce linked orderings and prove that there is always an optimum ordering that is linked.

\begin{definition}[linked ordering]
	An ordering $\sigma$ of $V(G)$ is \emph{linked} if  for any two vertices $u\leq_{\sigma} u'$, there exist $\min \{\delta(V^{\sigma}_{v}) \mid u\leq_{\sigma} v\leq_{\sigma} u' \}$ edge-disjoint paths between $V^{\sigma}_{u}$ and $V(G)\setminus V^{\sigma}_{u'}$ in $G$.
\end{definition}

\begin{lemma}[\!\!\cite{GeelenGW02,KanteK14}]\label{lem:linked}
For each graph $G$, there is a linked ordering $\sigma$ of $V(G)$ with $\ctw_{\sigma}(G)=\ctw(G)$.
\end{lemma}
\begin{proof}
Without loss of generality, we may assume that the graph is connected.
Let $\sigma$ be an optimum cutwidth ordering of $V=V(G)$. Subject to the optimality of $\sigma$, we choose $\sigma$ so that $\sum_{v\in V} \delta(V^{\sigma}_v)$ is minimized.
We prove that $\sigma$ defined in this manner is in fact linked.

Assume the contrary. Then by Menger's theorem, there exist vertices $u <_\sigma u'$ in $V$ and $i\in \mathbb{N}$ such that
$\delta(V^{\sigma}_{v})> i$ for every $u\leq_{\sigma} v\leq_{\sigma} u'$, but a minimum cut $(A,B)$ of $G$ with $V^{\sigma}_{u}\subseteq A$ and $V\setminus V^{\sigma}_{u'}\subseteq B$ has size $\delta(A) \leq i$.
We partition $A$ into sets $A_{1}$ and $A_{2}$, where $A_{1}=V^{\sigma}_{u}$ and $A_2=A\setminus A_1$, and we partition $B$ into sets $B_{1}$ and $B_{2}$, where
$B_{2}=V\setminus V^{\sigma}_{u'}$ and $B_1=B\setminus B_2$ (see Figure~\ref{fig:linked-proof}).
Notice that $A_{2}=A\setminus V^{\sigma}_{u}=\{v\mid u<_{\sigma} v \leq_{\sigma}u'\}\cap A$ and that 
$B_{1}=B\setminus (V\setminus V^{\sigma}_{u'})=\{v\mid u<_{\sigma} v \leq_{\sigma}u'\}\cap B$.
Let $\sigma'$ be the ordering of $V$ obtained by concatenating 
$\sigma|_{A_{1}}$, $\sigma|_{A_{2}}$, $\sigma|_{B_{1}}$, and $\sigma|_{B_{2}}$. 

We prove that $\delta(V^{\sigma'}_{v})\leq \delta(V^{\sigma}_{v})$, for every $v\in V$.
Observe first that for every vertex $v\in A_{1}\cup B_{2}$ it holds that $V^{\sigma'}_{v}=V^{\sigma}_{v}$ and thus, $\delta(V^{\sigma'}_{v})= \delta(V^{\sigma}_{v})$.
Let now $v\in A_{2}$. Then $V^{\sigma'}_{v}=V^{\sigma}_{v}\cap A$. By the submodularity of cuts it follows that
$\delta(V^{\sigma}_{v}\cup A)+\delta(V^{\sigma}_{v}\cap A)\leq \delta(A)+\delta(V^{\sigma}_{v})$.
Notice that 
$(V^{\sigma}_{v}\cup A,V\setminus(V^{\sigma}_{v}\cup A))$ is also a cut separating $A_{1}=V^{\sigma}_{u}$ and $B_{2}=V\setminus V^{\sigma}_{u'}$. From the minimality of $(A,B)$ it follows that
$\delta(A)\leq \delta(V^{\sigma}_{v}\cup A)$. Therefore, $\delta(V^{\sigma}_{v}\cap A)\leq \delta(V^{\sigma}_{v})$. As $V^{\sigma'}_{v}=V^{\sigma}_{v}\cap A$, 
we obtain that $\delta(V^{\sigma'}_{v})\leq \delta(V^{\sigma}_{v})$.

Symmetrically, let now $v\in B_{1}$. Then $V^{\sigma'}_{v}=V^{\sigma}_{v}\cup A$. By the submodularity of cuts we have
$\delta(V^{\sigma}_{v}\cup A)+\delta(V^{\sigma}_{v}\cap A)\leq \delta(A)+\delta(V^{\sigma}_{v})$.
Notice that $(V^{\sigma}_{v}\cap A,V\setminus(V^{\sigma}_{v}\cap A))$ is a cut separating $A_{1}$ and $B_{2}$. From the minimality of $(A,B)$ it follows that
$\delta(A)\leq \delta(V^{\sigma}_{v}\cap A)$. Therefore, $\delta(V^{\sigma}_{v}\cup A)\leq \delta(V^{\sigma}_{v})$. As $V^{\sigma'}_{v}=V^{\sigma}_{v}\cup A$, 
we obtain that $\delta(V^{\sigma'}_{v})\leq \delta(V^{\sigma}_{v})$.

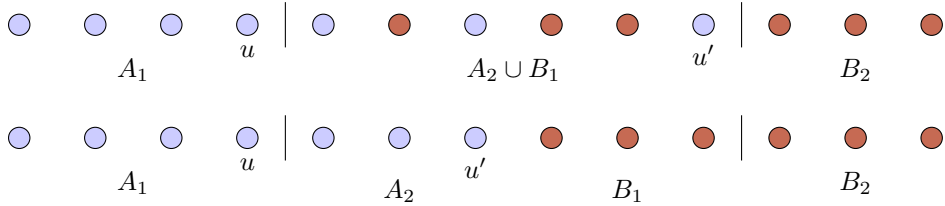
\begin{figure}[b]
	\centering
	
\begin{tikzpicture}[scale =1]
	\node[A] (v1) at (1,0) {};
	\node[A] (v2) at (2,0) {};
	\node[A] (v3) at (3,0) {};
	\node[A,label=below:$u$] (v4) at (4,0) {};
	\draw (4.5,-0.3)--(4.5,0.3);
	\node[A] (v5) at (5,0) {};
	\node[B] (v6) at (6,0) {};
	\node[A] (v7) at (7,0) {};
	\node[B] (v8) at (8,0) {};
	\node[B] (v9) at (9,0) {};
	\node[A,label=below:$u'$] (v10) at (10,0) {};
	\draw (10.5,-0.3)--(10.5,0.3);
	\node[B] (v11) at (11,0) {};
	\node[B] (v12) at (12,0) {};
	\node[B] (v13) at (13,0) {};

	\node (A1) at (2.5,-0.6) {$A_1$};
	\node (A2) at (7.5,-0.6) {$A_2 \cup B_1$};
	\node (B2) at (12,-0.6) {$B_2$};
	
	\begin{scope}[shift={(0,-1.5)}]
	\node[A] (v1) at (1,0) {};
	\node[A] (v2) at (2,0) {};
	\node[A] (v3) at (3,0) {};
	\node[A,label=below:$u$] (v4) at (4,0) {};
	\draw (4.5,-0.3)--(4.5,0.3);
	\node[A] (v5) at (5,0) {};
	\node[B] (v6) at (8,0) {};
	\node[A] (v7) at (6,0) {};
	\node[B] (v8) at (9,0) {};
	\node[B] (v9) at (10,0) {};
	\node[A,label=below:$u'$] (v7) at (7,0) {};
	\draw (10.5,-0.3)--(10.5,0.3);
	\node[B] (v11) at (11,0) {};
	\node[B] (v12) at (12,0) {};
	\node[B] (v13) at (13,0) {};

	\node (A1) at (2.5,-0.6) {$A_1$};
	\node (A2) at (6,-0.7) {$A_2$};
	\node (B1) at (9,-0.7) {$B_1$};
	\node (B2) at (12,-0.6) {$B_2$};
	\end{scope}
\end{tikzpicture}

	\caption{An ordering of vertices with the minimum cut $(A,B)$ between $A_1$ and $B_2$ of size $i$ highlighted in blue and red. Below, the modified ordering, with cutwidth bounded using submodularity.}
\label{fig:linked-proof}
\end{figure}

Thus, $\delta(V^{\sigma'}_{v})\leq \delta(V^{\sigma}_{v})\leq\ctw(G)$ for every $v\in V$, and hence $\ctw_{\sigma'}(G)=\ctw(G)$.
Finally, note that $\delta(V^{\sigma'}_{v}) = \delta(A) \leq i < \delta(V^{\sigma}_{v})$ for the last vertex $v$ in $A$. Thus $\sum_v \delta(V^{\sigma'}_v)<\sum_v \delta(V^{\sigma}_v)$, contradicting the choice of $\sigma$.
Therefore, $\sigma$ is a linked ordering of $V$ with $\ctw_{\sigma}(G)=\ctw(G)$.
\end{proof}

The rest of Section~\ref{sec:linked} is devoted 
to the proof of Theorem~\ref{thm:main}.
Before we proceed with this proof, we need a series of auxiliary lemmas.

 For every $s,r\in \mathbb{N}^+$, we set  $A_{s,r}=[s,s+r-1]$. We prove the following.
 
\begin{lemma}\label{lem:word_unpump}
	Let $N$ be a positive integer. 
	For every $s,r\in \mathbb{N}^+$ and every word $w$ over $A_{s,r}$ of length $N^{r}$
	there is a symbol $k\in A_{s,r}$ and a subword $u$ of $w$ such that (a) $u$ contains only numbers not smaller than $k$, and (b) $u$ contains the number $k$ at least $N$ 
	times.
\end{lemma}

\begin{proof} 
	We prove the lemma by induction on $r$. 
	Notice that for $r=1$, $A_{s,r}=\{s\}$ and thus the only word $w$ of length $N$ is $s^{N}$. Thus, the lemma holds with $k=s$ and $u=w$.
	We proceed to the inductive step for $r>1$.
	
	Let now $s\in \mathbb{N}$ and let $w$ be a word over $A_{s,r}$ of length $N^{r}$. If $s$ occurs at least $N$ times, then again, the lemma holds with $k=s$ and $u=w$.
	Thus, we may assume that $s$ occurs at most $N-1$ times. Then, since $w$ has length at least $N^{r}$, there
	exists a subword $w'$ of $w$ of length at least $N^{r-1}$ over $A_{s,r}\setminus\{s\} = A_{s+1,r-1}$.
	From the inductive hypothesis, there exists $k \in A_{s+1,r-1}\subseteq A_{s,r}$ and a 
	subword $u$ of $w'$ such that $k$ occurs at least $N$ times in $u$ and $u$ contains only numbers at least $k$. Since $w'$ is a subword of $w$, $u$ is also a
	subword of $w$. This completes the inductive step and the proof of the lemma.
\end{proof}

We use Lemma~\ref{lem:word_unpump} only for $s=1$, giving the following corollary.

\begin{corollary}\label{lem:wordlngth}
	Let $r,N$ be positive integers and let $w$ be a word of length $N^{r}$ over the alphabet $[r]$. Then there is a number $k\in [r]$ and a 
	subword $u$ of $w$ such that (a) $u$ contains only numbers not smaller than $k$, and (b) $u$ contains the number $k$ at least $N$ times.
\end{corollary}

We also need one additional statement about boundaried graphs and bucket interfaces.

\begin{lemma}\label{lem:subset_bucket}
Let $k,\ell\in\mathbb{N}$.
Suppose $\mathbf{A}=(A,\bar{x})$ and $\mathbf{B}=(B,\bar{y})$ are two $k$-boundaried graphs, 
and suppose further that there is an immersion model $(\phi,\psi)$ of $A$ in $B$ such that $\phi(x_i)=y_i$, for all $i=1,2,\ldots,k$.
Then for every $(k,\ell)$-bucket interface $(b,b',\mu,\mu^*)$, if $\mathbf{B}$ conforms to $(b,b',\mu,\mu^*)$ then also $\mathbf{A}$ conforms to $(b,b',\mu,\mu^*)$. 
\end{lemma}

\begin{proof}
First, we extend the immersion model $(\phi,\psi)$ to an immersion model $(\phi^*,\psi^*)$ of $A^*$ in $B^*$ by putting $\phi^*(x_i')=y_i'$ and $\psi^*(x_ix_i')=y_iy_i'$ for all $i\in [k]$.
Suppose that ordering $\sigma$ of $V(B^*)$ and its $\ell$-bucketing $T$ certify that $\mathbf{B}$ conforms to $(b,b',\mu,\mu^*)$.
We define ordering $\sigma'$ of $V(A^*)$ and its $\ell$-bucketing $T'$ as follows:
\begin{itemize}
\item For $u,v\in V(A^*)$, we put $u<_{\sigma'}v$ if and only if $\phi^*(u)<_{\sigma'}\phi^*(v)$.
\item For $u\in V(A^*)$, we put $T'(u)=T(\phi^*(u))$.
\end{itemize}
It is easy to see that $T'$ is an $\ell$-bucketing of $\sigma'$.
We now verify that $\sigma'$ and $T'$ certify that $\mathbf{A}$ conforms to $(b,b',\mu,\mu^*)$.
The first two conditions of conforming follow directly from the definition of $\sigma'$ and $T'$, so we are left with the third and the fourth condition.

For the third condition, take any $j\in [\ell]$. It suffices to show that for any cut $(L',R')\in \cuts(A,\sigma'|_A,T'|_A,j)$, we have that $|E_A(L',R')|\leq \mu(j)$.
By the construction of $(\sigma',T')$ it follows that there is a cut $(L,R)\in \cuts(B,\sigma|_B,T|_B,j)$ such that $\phi(L')\subseteq L$ and $\phi(R')\subseteq R$.
Since $(\sigma,T)$ certified that $\mathbf{B}$ conforms to $(b,b',\mu,\mu^*)$, we have that $|E_B(L,R)|\leq \mu(j)$.
Take any $uv\in E_A(L',R')$, and observe that $\psi(uv)$ is a path in $B$ leading from $\phi(u)\in L$ to $\phi(v)\in R$.
Consequently, one of the edges of this path must belong to $E_B(L,R)$.
Since paths $\psi(uv)$ are pairwise edge-disjoint for different edges $uv\in E_A(L',R')$, we infer that 
$$|E_A(L',R')|\leq |E_B(L,R)|\leq \mu(j).$$
This establishes the third condition. The fourth condition follows by the same argument applied to graphs $A^*$ and $B^*$, instead of $A$ and $B$.
\end{proof}

The following theorem is the technical counterpart of Theorem~\ref{thm:main}. Its proof is based 
on Theorem~\ref{thm:myhill_nerode}, Lemma~\ref{lem:linked}, Observation~\ref{obs:bcketsze} and the idea of ``unpumping'' repeating types, presented in the introduction.
The linkedness is used to make sure that within the unpumped segment of the ordering, one can find the maximum possible number of edge-disjoint paths between the parts of the graph on the left side and on the right side
of the segment. This ensures that the graph obtained from unpumping can be immersed in the original one.

\begin{theorem}\label{restated}
Let $k$ be a positive integer. If $G\in \obs_{\leq_{\rm si}}(\C_{k})$, then $|V(G)|\leq N^{k+1}$, where $N=2^{2((k+1)\log\ell +\ell\log(k+2))}+2$ and $\ell=(2k+3) \cdot (2k+6)$.
\end{theorem}

\begin{proof}
	Take any $G\in \obs_{\leq_{si}}(\C_{k})$ and assume, towards a contradiction,  that $|V(G)|> N^{k+1}$. Let
	$\sigma=\langle v_{1},v_{2},\dots, v_{|V(G)|}\rangle$ be a linked optimum cutwidth ordering of $G$, which exists by Lemma~\ref{lem:linked}. 
	We define $c_{i}=\delta(V^{\sigma}_{v_{i}})$, that is, $c_{i}$ is the size of
	the cut between the vertices of $G$ up to $v_{i}$ and the rest of the graph. Notice that since $G\in\obs_{\leq_{si}}(\C_{k})$, we have that $\ctw(G)=k+1$ and $G$ is connected.
	This implies that $c_{i}\in [k+1]$, for every $i\in [|V(G)|-1]$. 
	
	Observe that $c_{1}c_{2}\dots c_{|V(G)|-1}$ is a word of length at least $N^{k+1}$ over the
	alphabet $[k+1]$. From Corollary~\ref{lem:wordlngth}, it follows that there exist $1\leq s\leq t< |V(G)|$ and $q\in [k+1]$ such that for every $s\leq i\leq t$ we have $c_{i}\geq q$,
	and there also exist $N$ distinct indices $s\leq i_{1}<i_{2}<\dots<i_{N}\leq t$ such that $c_{i_j}=q$, for every $j\in [N]$. Without loss of generality
	we may assume that $i_{1}=s$ and $i_{N}=t$.

	For each $j\in [N]$, we can define a $q$-boundaried graph $\mathbf{G}_{j}=(G_{j},(z_{j}^{1},z_{j}^{2},\dots,z_{j}^{q}))$ in the following way.
	First, by linkedness, we find edge-disjoint paths $P_{1}, \dots, P_{q}$ between $V^\sigma_{v_{i_{1}}}$ and $V\setminus V^\sigma_{v_{i_{N}}}$.
	Notice that for each $j\in [N]$ the cut $E_G(V^{\sigma}_{v_{i_j}},V(G)\setminus V^{\sigma}_{v_{i_j}})$ contains exactly one edge of each path $P_{i}$. Denote this edge by $e_{j}^{i}$, for $i\in [q]$.
	For $i\in [q]$, let $x_j^i$ be the endpoint of $e_{j}^{i}$ that belongs to $V^{\sigma}_{v_{i_j}}$, and let $y_j^i$ be the endpoint that does not belong to $V^{\sigma}_{v_{i_j}}$.
	We construct $G_j$ by taking $G[V^{\sigma}_{v_{i_j}}]$, adding fresh boundary vertices $(z_{j}^{1},z_{j}^{2},\dots,z_{j}^{q})$, and adding one fresh edge $x_j^iz_j^i$ for each $i\in [q]$.
	
	Now consider any pair of indices $1\leq j_1<j_2\leq N$. Observe that there exists an immersion model $(\phi,\psi)$ of $\mathbf{G}_{j_1}$ in $\mathbf{G}_{j_2}$ such that
	$\phi(z_{j_1}^i)=z_{j_2}^i$ for each $i\in [q]$. Indeed, we can put $\phi(u)=u$ for each $u\in V(G_{j_1})$ and $\phi(z_{j_1}^i)=z_{j_2}^i$ for each $i\in [q]$.
	Then $\psi$ can be defined by taking $\psi(e)=e$ for each $e\in E(G_{j_1})$ and mapping each edge $x_{j_1}^iz_{j_1}^i$ to an appropriate infix of the path $P_i$, extended by the edge
	$x_{j_2}^iz_{j_2}^i$. Consequently, $\mathbf{G}_{j_1}$ and $\mathbf{G}_{j_2}$ satisfy the prerequisites of Lemma~\ref{lem:subset_bucket}.
	We infer that if by $\zeta(j)$ we denote the set of $(q,\ell$)-bucket interfaces to which $\mathbf{G}_j$ conforms, then
	$$\zeta(1)\supseteq \zeta(2)\supseteq \ldots\supseteq \zeta(N-1)\supseteq \zeta(N).$$
	Observation~\ref{obs:bcketsze} implies that $N$ is larger by more than $1$ than the total number of $(q,\ell)$-bucket interfaces.
	It follows that there exists an index $j$, $1\leq j<N$, such that
	$$\zeta(j)=\zeta(j+1).$$
	In other words, the $q$-boundaried graphs $\mathbf{G}_j$ and $\mathbf{G}_{j+1}$ are $(q,\ell)$-similar.
	
	Define a $q$-boundaried graph $\mathbf{G'}=(G',(y_{j+1}^1,\ldots,y_{j+1}^q))$ by taking $G'=G[V(G)\setminus V^\sigma_{i_{j+1}}]$.
	It can be now seen that $\mathbf{G}_{j+1}\oplus\mathbf{G'}$ is exactly the graph $G$ with every edge of the cut $E_G(V^{\sigma}_{v_{i_j}},V(G)\setminus V^{\sigma}_{v_{i_j}})$ subdivided once.
	Since subdividing edges does not change the cutwidth of the graph, we have that
	\begin{equation}\label{eq:subdivide}
	\ctw(\mathbf{G}_{j+1}\oplus\mathbf{G'})=\ctw(G)>k.
	\end{equation}
	On the other hand, $q$-boundaried graphs $\mathbf{G}_j$ and $\mathbf{G}_{j+1}$ are $(q,\ell)$-similar. 
	Since $\ell\geq (2q+3) \cdot (2q+6)$, by Theorem~\ref{thm:myhill_nerode} we conclude that
	\begin{equation}\label{eq:myhill}
	\ctw(\mathbf{G}_{j}\oplus \mathbf{G}')=\ctw(\mathbf{G}_{j+1}\oplus \mathbf{G}').
	\end{equation}
	Examine the graph $\mathbf{G}_{j}\oplus \mathbf{G}'$. In the join operation, we added an edge $z_j^iy_{j+1}^i$ for each $i\in [q]$, which means each vertex $z_j^i$
	has exactly two incident edges in $\mathbf{G}_{j}\oplus \mathbf{G}'$: one connecting it to $x_j^i$ and one connecting it to $y_{j+1}^i$.
	Let $H$ be the graph obtained from $\mathbf{G}_{j}\oplus \mathbf{G}'$ by dissolving every vertex $z_j^i$, i.e., removing it and replacing edges $x_{j}^iz_j^i$ and $z_j^iy_{j+1}^i$ with a fresh
	edge $x_j^iy_{j+1}^i$. Subdividing edges does not change the cutwidth of a graph, so we obtain that:
	\begin{equation}\label{eq:dissolve}
	\ctw(H)=\ctw(\mathbf{G}_{j}\oplus \mathbf{G}')
	\end{equation}
	Finally, it is easy to see that $G$ admits $H$ as a strong immersion: 
	a strong immersion model of $H$ in $G$ can be constructed by mapping the vertices and edges of $G_j$ and $G'$ identically, and then mapping each of the remaining edges $x_j^iy_{j+1}^i$ to a corresponding infix of
	the path $P_i$. Also, since $i_j<i_{j+1}$, the graph $H$ has strictly less vertices than $G$. However, from Eq.~\eqref{eq:subdivide},~\eqref{eq:myhill}, and~\eqref{eq:dissolve} we conclude that
	$\ctw(H)=\ctw(G)>k$. This contradicts the assumption that $G\in\obs_{\leq_{si}}(\C_{k})$.
	\end{proof}

\begin{proof}[Proof of Theorem~\ref{thm:main}]
Theorem~\ref{restated} provides an upper bound on the number of vertices of a graph in $\obs_{\leq_{si}}(\C_{k})$. 
Observe that since such a graph has cutwidth $k+1$, each of its vertices has degree at most $2(k+1)$.
It follows that any graph from $\obs_{\leq_{si}}(\C_{k})$ has $2^{\Oh(k^3\log k)}$ vertices and edges.
Finally, by Observation~\ref{osesimme} we have ${\bf obs}_{\leq_{\rm i}}(\C_{q})\subseteq  {\bf obs}_{\leq_{si}}(\C_{q})$, so
the same bound holds also for immersions instead of strong immersions. This concludes the proof of Theorem~\ref{thm:main}.	 
	\end{proof}

\section{An algorithm for computing cutwidth}\label{sec:algo}

In this section we present an exact FPT algorithm for computing the cutwidth of the graph.
First, we need to give a dynamic programming algorithm that given an approximate ordering $\sigma$ of width $r$, finds, if possible,
an ordering of width at most $k$, where $k\leq r$ is given.

Our algorithm takes advantage of the given ordering $\sigma$ and essentially computes, for each subgraph of $G$ induced by a prefix of $\sigma$, the $(r,\ell)$-bucket interfaces it conforms to.
More precisely, in Lemma~\ref{lem:dpalgcor} we show that if $G$ has
an optimum ordering of width $k$, then there is an optimum ordering were \emph{each} of these induced subgraphs
occupies at most $\ell = O(rk)$ buckets, allowing to restrict our search to $(r,\ell)$-bucket profiles (a variant of bucket interfaces to be defined later, refined so as to consider border vertices more precisely). 
The proof slightly strengthens that of Lemma~\ref{lem:num_blocks}.

\begin{lemma}\label{lem:dpalgcor}
Let $G$ be a graph with an ordering $\sigma$ of width $r$. 
Then there exists also an ordering $\tau$ of optimum width, i.e., with $\ctw_\tau(G)=\ctw(G)$,
that has the following property: for every prefix $X$ of $\sigma$, the number of $X$-blocks in $\tau$ is at most $2r\cdot \ctw(G)+\ctw(G)+4r+2$.
\end{lemma}
\begin{proof}
Lemma~\ref{lem:num_blocks} asserts that for each cut $(A,B)$ of $G$ of size at most $r$, there exists an optimum-width ordering of $V(G)$ where the number of $(A,B)$-blocks is at most 
$$(2r+1) \cdot (2\ctw(G)+3)+2r=4r\cdot \ctw(G)+2\ctw(G)+8r+3.$$
As $A$-blocks and $B$-blocks appear alternately, at most half rounded up of the $(A,B)$-blocks can be $A$-blocks. 
Hence, the number of $A$-blocks in such an optimum-width ordering is at most $2r\cdot \ctw(G)+\ctw(G)+4r+2$; we denote this quantity by $\lambda$.

The proof of Lemma~\ref{lem:num_blocks} in fact shows that for any ordering $\sigma$ of $V(G)$ and any cut $(A,B)$ of $G$ of size at most $r$, either $\sigma$ already has at most $2\lambda-1$ $(A,B)$-blocks, or an ordering $\sigma'$ can be obtained from $\sigma$ by swapping its $(A,B)$-blocks so that $\sigma'$ has strictly less $(A,B)$-blocks.
Therefore, by reordering $(A,B)$-blocks of $\sigma$, we eventually get a new ordering which has at most $2\lambda-1$ $(A,B)$-blocks, and hence at most $\lambda$ $A$-blocks.

For $i=1,2,\ldots,|V(G)|-1$, let $(A_i,B_i)$ be the cut of $G$, where $A_i$ is the prefix of $\sigma$ of length $i$, while $B_i$ is the suffix of $\sigma$ of length $|V(G)|-i$.
Let $\tau_0$ be any optimum-width ordering of $G$. We now inductively construct orderings $\tau_1,\tau_2,\ldots,\tau_{|V(G)|-1}$, as follows: once $\tau_i$ is constructed, we apply the above reordering procedure to $\tau_i$ and cut $(A_{i+1},B_{i+1})$. 
This yields a new ordering $\tau_{i+1}$ of optimum width such that the number of $A_{i+1}$-blocks in $\tau_{i+1}$ is at most $\lambda$.
Furthermore, $\tau_{i+1}$ is obtained from $\tau_i$ by reordering $A_{i+1}$- and $B_{i+1}$-blocks in $\tau_i$.
Hence, whenever $X$ is a subset of $A_{i+1}$, then
any $X$-block in $\tau_i$ remains consecutive in $\tau_{i+1}$, as it is contained in one $A_{i+1}$-block in $\tau_i$ that is moved as a whole in the construction of $\tau_{i+1}$.
Consequently, if for all $j\leq i$ we had that the number of $A_j$-blocks in $\tau_i$ is at most $\lambda$, then this property is also satisfied in $\tau_{i+1}$.
It is now clear that a straightforward induction yields the following invariant: for each $j\leq i$, then number of $A_j$-blocks in $\tau_i$ is at most $\lambda$.
Therefore $\tau=\tau_{|V(G)|-1}$ gives an ordering with the claimed properties.
\end{proof}

\paragraph{Bucket profiles.}
We now define a refinement of the widths of the buckets of a bucket interface as well as a refinement of the notion
of bucket interfaces. They are  used in the dynamic programming algorithm  of Lemma~\ref{lem:dynamicCtw}.
\begin{definition}
	Let $(G,\bar{x})$ be a $k$-boundaried graph and let $S = \{x_1,\dots,x_k,x_1',\dots,x_k'\}\subseteq V(G^*)$. 
	Let now $\sigma$ be an ordering of $V(G^*)$ and $T$ be an $\ell$-bucketing of $\sigma$.
	For every bucket $T^{-1}(i)$, $i \in [\ell]$, let $T^{-1}(i)\cap S=\{v_{1},v_{2},\dots,v_{p}\}$ for some $v_{1}<_{\sigma} v_{2}<_{\sigma} \dots <_{\sigma} v_{p}$; we then define 
$$T^{-1}_{j}(i)  = 
	\begin{cases} 
	 \{v\in T^{-1}(i)\colon v<_{\sigma} v_{1} \} & \mbox{for } j=0,\\
	 \{v\in T^{-1}(i)\colon v_{j}<_{\sigma} v<_{\sigma} v_{j+1} \} & \mbox{for } j\in [p-1],\\
	 \{v\in T^{-1}(i)\colon v_{p}<_{\sigma} v \} & \mbox{for } j=p.
	 \end{cases}$$
	Let also $\cuts(G,\sigma,T,i,j)$ be the family of $\sigma$-cuts containing on one side all vertices 
	appearing before $v_{j-1}$ (or, if $j=0$, all vertices of buckets appearing before 
	bucket $i$) and a prefix (in $\sigma$) of $T^{-1}_{j}(i)$. For an ordering $\sigma$ of the vertices of a graph $G$, 
	define the \emph{width of $j$-th segment} $T^{-1}_{j}(i)$ of the bucket $i$, $i\in[\ell]$, $j\in [0,p]$, 
	as the maximum width of any cut in the family $\cuts(G,\sigma,T,i,j)$.
	Formally,
	\begin{eqnarray*}
	\cuts(G,\sigma,T,i,j) & = & \left\{ \left(T^{-1}(\{1,\dots,i-1\})\cup L, T^{-1}(\{i+1,\dots,\ell\})\cup R\right) \colon \right. \\
	& &~~\left. (L,R)\mbox{ is a }\sigma\mbox{-cut of }T^{-1}(i)\mbox{ with }v_j\in L\mbox{ and }v_{j+1}\in R\right\},\\
	\width(G,\sigma,T,i,j) & = & \max \left\{\, \left|E_G(L,R)\right|\ \colon\ (L,R) \in \cuts(G,\sigma,T,i,j)\, \right\}.
	\end{eqnarray*}
\end{definition}	%

\noindent We also need to refine the notion of a $(k,\ell)$-bucket interface. 
\begin{definition}\label{def:profile}
	For $k,\ell\in\mathbb{N}$, a \emph{($k$,$\ell$)-bucket profile} consists of functions:
	\begin{itemize}
		\item $b,b': [k] \to [\ell]$ identifying the buckets which contain $x_{i}$ and $x_{i}'$, respectively,
		\item $p,p':[k]\to [k]$ highlighting the ordering between the vertices $x_{i}$ and $x_{i}'$ inside a bucket, respectively,
		\item $\nu: [\ell]\times [0,k] \to [0,k]$ corresponding to the widths of segments of buckets defined by the vertices $x_{i}$, respectively.
	\end{itemize}

	A $k$-boundaried graph $\mathbf{G}$ \emph{conforms} with a $(k,\ell)$-bucket profile,
	if there exists an ordering $\sigma$ of the vertices of $G^{*}$ and an $\ell$-bucketing $T$ such that:

	\begin{itemize}
		\item $T(v)$ is odd for $v\in V(G)$ and even for $v\in \{x_{1}',\dots,x_k'\}$,
		\item $T(x_{i}) = b(i)$ and $T(x_{i}') = b'(i)$, for each $i\in [k]$,
		\item $p(i)<p(j)$, if $b(i)=b(j)$ and $x_{i}<_{\sigma} x_{j}$, and $p'(i)<p'(j)$ if $b'(i)=b'(j)$ and $x_{i}'<_{\sigma} x_{j}'$,
		\item $\width(G, \sigma|_{G}, T|_{G}, j,s) = \nu(j,s)$, for each $j\in [\ell]$ and $s\in [0,k]$.
	\end{itemize}
\end{definition}

From the fact that the boundary vertices of a $k$-boundaried graph $\mathbf{G}$ split the buckets defined by $T$ into at most $2k$ segments in total it follows that:

\begin{observation}\label{obs:prflsze}
For any pair $(k,\ell)$ of positive integers, there is a set of at most $$2^{2k(\log\ell +\log k)+(\ell+2k)\log(k+1)}$$ $(k,\ell)$-bucket profiles that
a $k$-boundaried graph $\mathbf{G}$ can possibly conform with, and this set can be constructed in time polynomial in its size.
\end{observation}

The $(k,\ell)$-bucket profiles that Observation~\ref{obs:prflsze} refers to will be called {\em{valid}}.
By making use of these two notions we ensure that we will be able to update the widths of each bucket every time 
a new vertex is processed by the dynamic programming algorithm. We are now ready to prove Lemma~\ref{lem:dynamicCtw}.

\begin{lemma}\label{lem:dynamicCtw} Let $r\in\NN^+$.
	Given a graph $G$ and an ordering $\sigma$ of its vertices with $\ctw_{\sigma}(G)\leq r$, 
	an ordering $\tau$ of the vertices of $G$ with $\ctw_{\tau}(G)=\ctw(G)$ can be computed 
	in time $2^{\Oh(r^2 \log r)} \cdot |V(G)|$.
\end{lemma}

\begin{proof}
The algorithm attempts to compute an ordering of width $k$ for consecutive $k=0,1,2,\ldots$. 
The first value of $k$ for which the algorithms succeeds is equal to the value of the cutwidth, and then the constructed ordering may be
returned. Since there is an ordering of width $r$, we will always eventually succeed for some $k\leq r$, which implies that we will make at most $r+1$ iterations.
Hence, from now on we may assume that we know the target width $k\leq r$ for which we try to construct an ordering.

Given a graph $G$ and an ordering $\sigma$ of its vertices with $\ctw_{\sigma}(G)\leq r$ we
denote by $G_{w}$ the graph induced by the vertices of the prefix of $\sigma$ of length $w$.
Then we naturally define the boundaried graph $\mathbf{G}_{w}$, where we introduce a boundary vertex $x_i$ 
for each edge $e_i$ of the cut $E_G(V(G_{w}),V(G)\setminus V(G_{w})$. Note that this cut has at most $r$ edges.

By Lemma~\ref{lem:dpalgcor}, we know that there is an optimum-width ordering $\tau$ such that every prefix $V(G_w)$ of $\sigma$ has at most $\ell$ blocks in $\tau$.
Our dynamic programming algorithm will simply inductively reconstruct all $(k,\ell)$-bucket profiles that may correspond to $V(G_w)$-blocks in $\tau$, for each consecutive $w$ in the ordering $\sigma$, eventually reconstructing $\tau$, if $\ctw_\tau(G)\leq k$.

We now construct an auxiliary directed graph $D$ that will model states and transitions of our dynamic programming algorithm.
Let $\ell=4rk+2k+8r+4$.
First, for every $w\in [0,|V(G)|]$ and every valid $(k,\ell)$-bucket profile $P$, we add
a vertex $(w,P)$ to $D$. 
Thus, by Observation~\ref{obs:prflsze}, the digraph $D$ has at most $$2^{2k(\log\ell+\log k)+(\ell+2k)\log(k+1)}\cdot (|V(G)|+1)=2^{\Oh(r^{2}\log r)}\cdot |V(G)|$$ vertices.
We add an edge  $((w,P),(w+1,P'))$, whenever the $(k,\ell)$-bucket profile 
$P$ can be {\em{expanded}} to the $(k,\ell)$-bucket profile $P'$ in the sense that we explain now.

We describe which bucket profiles $P'$ expand $P$ by guessing where the new vertex would land in the bucket profile $P$, assuming that $\mathbf{G}_w$ conforms to $P$. 
After the guess is made, the updated profile $P$ becomes the expanded profile $P'$.
Different guesses lead to different profiles $P'$ which extend $P$; this corresponds to different ways in which the construction of the optimum ordering can continue.
As describing the details of this expansion relation is a routine task, we prefer to keep the description rather informal, and leave working out all the formal details to the reader.

Let $v_{w+1}$ be the $(w+1)$-st vertex in the ordering $\sigma$, that is, $v_{w+1}\in V(G_{w+1})\setminus V(G_{w})$.
We construct (by guessing) a $(k,\ell)$-bucket profile $P'$ from the $(k,\ell)$-bucket profile $P$ in the following
way.
First, we guess an even bucket of $P$ to place each one of the vertices in $V(G_{w+1}^{*})\setminus V(G_{w}^{*})$: 
the new vertices of the extension that correspond to new edges of the cut $E_G(V(G_{w+1}),V(G)\setminus V(G_{w+1}))$ that are incident to $v_{w+1}$. 
Notice that each bucket contains, at any moment, at most $r$ vertices. Therefore, we have at most $r+1$ possible choices
about where each vertex will land in each bucket (including the placing in the order, as indicated by the function $p'(\cdot)$.  
Notice also that there are at most $r+1$ vertices in $V(G_{w+1}^{*})\setminus V(G_{w}^{*})$. Therefore we have at most
$(\ell (r+1))^{r+1}$ options for this guess.

Next, we choose the place $v_{w+1}$ is going to be put in.
If $v_{w+1}$ is an endpoint of an edge from the cut $E_G(V(G_w),V(G)\setminus V(G_w))$, then this place is already indicated by functions $b'(\cdot)$ and $p'(\cdot)$ in the bucket profile $P$;
if there are multiple edges in the cut $E_G(V(G_w),V(G)\setminus V(G_w))$ that have $v_{w+1}$ as an endpoint, then all of them must be placed next to each other in the same even bucket (otherwise $P$ has no extension).
Otherwise, if $v_{w+1}$ is not an endpoint of an edge from $E_G(V(G_w),V(G)\setminus V(G_w))$, we guess the placing of $v_{w+1}$ by guessing
an even bucket (one of at most $\ell+1$ options) together with a segment between two consecutive extension vertices in this bucket (one of at most $r+1$ options).

The placing of $v_{w+1}$ may lead to one of three different scenarios; we again guess which one applies.
First, $v_{w+1}$ can establish a new odd bucket and split the even bucket into which it was put into two new even buckets, one
on the left and one on the right of the new odd bucket containing $v_{w+1}$; the other extension vertices placed in this bucket are split accordingly.
Second, $v_{w+1}$ can be present at the leftmost or rightmost end of the even bucket it is placed in, so it gets merged into the neighboring odd bucket.
Finally, if the even bucket in which $v_{w+1}$ is placed did not contain any other extension vertices of $G_w^*$, then $v_{w+1}$ can be declared to be the last vertex placed in this bucket,
in which case we merge it together with both neighboring odd buckets.
In these scenarios, whenever the extended profile turns out to have more than $\ell$ buckets, we discard this option.

Having guessed how the placing of $v_{w+1}$ will affect the configuration of buckets, we proceed with updating the sizes of cuts, as indicated by the function $\nu(\cdot)$.
For this, we first examine all the edges of the cut $E_G(V(G_w),V(G)\setminus V(G_w))$ that have $v_{w+1}$ as an endpoint. These edges did not contribute to the values of $\nu(\cdot)$ in the bucket
profile $P$, but should contribute in $P'$. Note that given the placement of $v_{w+1}$, for each such edge we exactly see over which segments this edge ``flies over'', and therefore we can update the values
of $\nu(\cdot)$ for these segments by incrementing them by one. Finally, when $v_{w+1}$ got merged to a neighboring odd bucket (or to two of them), we may also need to take into account one more cut in
the value of $\nu(\cdot)$ for the last/first segment of this bucket: the one between $v_{w+1}$ and the vertices placed in this bucket. 
It is easy to see that from the value of $\nu(\cdot)$ for the segment
in which $v_{w+1}$ is placed, and the exact placement of the endpoints of all the boundary edges, we can deduce the exact size of this cut. Hence, the relevant value of $\nu(\cdot)$ can be efficiently updated by
taking the maximum of the old value and the deduced size of the cut. We update $\nu$ in a similar fashion when $v_{w+1}$ merges with both neighboring odd buckets.
If at any point any of the values of $\nu(\cdot)$ exceeds $k$, we discard this guess.

This concludes the definition of the extension. For every $(k,\ell)$-bucket profile $P$ and every $(k,\ell)$-bucket profile $P'$ that extends it, we add to $D$ an arc from $(w,P)$ to $(w+1,P')$.
It is easy to see from the description above that, given $P$ and $P'$, it can be verified in time polynomial in $r$ whether such an arc should be added.

Finally, in the graph $D$ we determine using, say, depth-first search, whether there is a directed path from node $(0,P_\emptyset)$ to node $(|V(G)|,P_{\rm full})$, where $P_\emptyset$ is an empty bucket profile
and $P_{\rm full}$ is a bucket profile containing just one odd bucket.
It is clear from the construction that if we find such a path, then by applying operations recorded along such a path we obtain an ordering of the vertices of $G$ of width at most $k$.
On the other hand, provided $k=\ctw(G)$, by Lemma~\ref{lem:dpalgcor} we know that there is always an optimum-width ordering $\tau$ such that every prefix of $\sigma$ has at most $\ell$ blocks in $\tau$.
Then the $(k,\ell)$-bucket profiles naturally defined by the prefixes of $\sigma$ in $\tau$ define a path from $(0,P_\emptyset)$ to $(|V(G)|,P_{\rm full})$ in $D$.

The graph $D$ has $2^{\Oh(r^{2}\log r)}\cdot|V(G)|$ vertices and arcs, and the depth-first search runs in time linear in its size.
It is also trivial to reconstruct the optimum-width ordering of the vertices of $G$ from the obtained path in linear time.
This yields the promised running time bounds.
\end{proof}

%

Having the algorithm of Lemma~\ref{lem:dynamicCtw}, a standard application of the iterative compression technique immediately yields a $2^{\Oh(k^2 \log k)} \cdot n^2$ time 
algorithm for computing cutwidth, as sketched in Section~\ref{sec:intro}.
Simply add the vertices of $G$ one by one, and apply the algorithm of Lemma~\ref{lem:dynamicCtw} at each step.
However, we can make the dependence on $n$ linear by adapting the approach of Bodlaender~\cite{Bodlaender96}; more precisely, we make bigger steps.
Such a big step consists of finding a graph $H$ that can be immersed in the input graph $G$, which is smaller by a constant fraction, but whose cutwidth is not much smaller.
This is formalised in Lemma~\ref{lem:reduce}. For its proof we 
we need the following definition and a known result about obstacles to small cutwidth.

\begin{definition}
A {\em perfect binary tree} is a rooted binary tree in which all interior nodes have two children and all leaves have the same distance from its root.
The {\em height} of a perfect binary tree is the distance between its root and one of its leaves.
\end{definition}

\begin{lemma}[\!\!\cite{Takahashi94-Mi,KORACH199397,GOVINDAN2001189}]\label{lem:ekjgherk}
If $T$ is a perfect binary tree of height $2k$, then $\ctw(T)\geq k$.
\end{lemma}


%
\begin{lemma}\label{lem:reduce}
	There is an algorithm that given a positive integer $k$ and a graph $G$, works in time $\Oh(k^2 \cdot |V(G)|)$ and either concludes that $\ctw(G) > k$, 
	or finds a graph $H$ immersed in $G$ such that $|E(H)| \leq |E(G)| \cdot (1-1/(2k+1)^{4(k+1)+3})$ and $\ctw(G) \leq 2\ctw(H)$.
	Furthermore, in the latter case, given an ordering $\sigma$ of the vertices of $H$, an ordering $\tau$ of the vertices of $G$ with $\ctw_{\tau}(G)\leq 2\ctw_{\sigma}(H)$ 
	can be computed in $\Oh(|V(G)|)$ time.
\end{lemma}

\begin{proof}
		Without loss of generality we assume that $G$ is connected, because we can apply the algorithm on the connected components of $G$ separately and then take the disjoint union of the results.

		Observe first that we may assume that every vertex in $G$ is incident to at most $2k$ edges, as otherwise, we could immediately conclude that $\ctw(G)>k$.
		This also implies that every vertex in $G$ has at most $2k$ neighbors; by $N(v)$ we denote the set of neighbors of a vertex $v$, 
		and $N(X)=(\bigcup_{v\in X} N(v))\setminus X$ for a vertex subset $X$.
		Let $G'$ be the graph obtained from $G$ by exhaustively dissolving any vertices of degree $2$ whose neighbors are different.
		That is, having such a vertex $v$, we delete it from the graph and replace the two edges incident to it with a fresh edge between its neighbors, and we proceed doing this as long as there are
		such vertices in the graph.
		Clearly, the eventually obtained graph $G'$ can be immersed in $G$, we have $|E(G')|\leq |E(G)|$, the degree of every vertex in $G'$ is the same to its degree in $G$, and $\ctw(G') \leq \ctw(G)$. 
		However, observe that any ordering of the vertices of $G'$ can be turned into an ordering of the vertices of $G$ with the same width by placing each dissolved vertex in any place 
		between its two original neighbors. Thus, $\ctw(G')=\ctw(G)$. 
		
		Moreover, $G'$ can be constructed in linear time by inspecting, in any order, all the vertices that have degree $2$ in the original graph $G$.
		It is also easy to see that, given an ordering of vertices of $G'$, one can reconstruct in linear time an ordering of $G$ of at most the same width.  
		
		Altogether, it is now enough to either conclude that $\ctw(G')>k$ or 
		find a graph $H$ immersed in $G'$ such that $$|E(H)| \leq |E(G')| \cdot (1-1/(2k+1)^{4(k+1)+2})$$ and $\ctw(G') \leq 2\ctw(H')$.
		Therefore, from now on we may assume that if the graph $G'$ contains a vertex that is incident to two edges then this vertex is incident to an edge of multiplicity 2.
		Let $V_{1}$ be the set of vertices of degree 1 in $G'$.
		We consider two cases depending on the size of $V_{1}$.\\

		\noindent {\bf Case 1.} $|V_{1}|  \geq |E(G')|/(2k+1)^{4(k+1)+2}$. 
		Notice first that $V_{1}\subseteq N(N(V_{1}))$, and recall that every vertex in $G'$ is incident to at most $2k$
		edges and therefore has at most $2k$ neighbors.
		It follows then that $|V_{1}| \leq 2k \cdot |N(V_{1})|$ and hence $|N(V_{1})| \geq |E(G')|/(2k+1)^{4(k+1)+3}$.
		Let $H$ be the graph obtained from $G'$ by removing, for each vertex in $N(V_{1})$, one of its neighbors in $V_1$.
		Then $|E(H)| \leq |E(G')| \cdot (1 - 1/(2k+1)^{4(k+1)+3})$ and $H$ is immersed in $G'$ (as it is an induced subgraph). Hence, $H$ is also immersed in $G$. 
		Furthermore, let $\sigma$ be any ordering of the vertices of $H$. Then, we can obtain an ordering of the vertices of $G'$ by placing each deleted vertex next to its original  		neighbors. 
		Notice that this placement increases the width of $\sigma$ by at most 1 in total, and thus by a multiplicative factor of at most 2.
		As we already showed how to obtain an ordering of $V(G)$ from a given ordering of $V(G')$, the lemma follows for the case where 
		$|V_{1}|  \geq |E(G')|/(2k+1)^{4(k+1)+2}$.\\
			
		\noindent {\bf Case 2.} $|V_{1}| < |E(G')|/(2k+1)^{4(k+1)+2}$.
		For every $v\in V(G')$ and every positive integer $s$, we define 
		$B_{s}(v)$ to be the ball of radius $s$ around $v$, that is, the set of vertices at distance at most $s$ from $v$ in $G'$.
		Recall that every vertex of $G'$ has at most $2k$ neighbors and observe then that $|B_{s}(v)| \leq (2k+1)^{s}$.
		We construct a set of vertices $v_{1}, v_{2}, \dots, v_{\ell} \in V(G')$ whose pairwise distance is greater than $4(k+1)$ in the following greedy way.
		Having chosen $v_1, \dots, v_i$, if $B_{4(k+1)}(v_1) \cup \dots \cup B_{4(k+1)}(v_{i})\neq V(G')$ then let $v_{i+1}$ be any vertex outside of 
		$B_{4(k+1)}(v_1) \cup \dots \cup B_{4(k+1)}(v_{i})$. If such a vertex does not exist, we stop by putting $\ell=i$ and consider the set $v_{1},v_{2},\dots,v_{\ell}$. 
		Observe here that we can calculate $B_{4(k+1)}(v_{i})$ by breadth-first search in $\Oh((2k+1)^{4(k+1)+1})$ time, by stopping the search at depth $4(k+1)$.
		However, note we do not need to revisit a previously visited vertex, unless we reach it with fewer steps. That is, starting with $i=0$, we mark which vertices we have already visited (the set $B_{4(k+1)}(v_1) \cup \dots \cup  B_{4(k+1)}(v_i)$) and remember minimum distances from $\{v_1,\dots,v_i\}$ to each previously visited vertex. Considering vertices in any order, we let $v_{i+1}$ be the first not yet visited. We then mark the new ball of radius $4(k+1)$ around it, but only exploring a previously visited vertex when the minimum distance to it strictly decreases by adding $v_{i+1}$. This way, we explore each vertex at most $4(k+1)$ times, as this is an upper bound on the minimum distance of any vertex when first visited. Hence the sequence $v_1,\dots,v_\ell$ can be computed in $\Oh(k^2 |V(G)|)$ time.
		We now estimate the length $\ell$ of the sequence. 
		
		Recall that for every $i\in[\ell]$, $|B_{4(k+1)}(v_i)|\leq (2k+1)^{4(k+1)}$ and that $V(G)=\bigcup_{i\in [\ell]} B_{4(k+1)}(v_{i})$.
		From the above and the fact that $|E(G')|\leq 2k\cdot |V(G')|$ (as every vertex of $G'$ is incident to at most $2k$ edges of $G'$),
		 it follows that $$\ell\geq |V(G')|/(2k+1)^{4(k+1)} \geq |E(G')|/(2k+1)^{4(k+1)+1}.$$

		By construction, the distance between $v_{i}$ and $v_{j}$ is greater than $4(k+1)$, for distinct $i,j \in [\ell]$.
		Therefore, the balls $B_{2(k+1)}(v_{1}), \dots, B_{2(k+1)}(v_{\ell})$ are vertex-disjoint.
		Moreover, since we have that $|V_{1}| < |E(G')|/(2k+1)^{4(k+1)+2}$, at most $|E(G')|/(2k+1)^{4(k+1)+2}$ of those balls contain a vertex of degree 1.
		Therefore, the remaining $\ell - |E(G')|/(2k+1)^{4(k+1)+2}$ balls are disjoint with $V_{1}$.
		Let $I\subseteq [\ell]$ be the set of indices for which the balls $B_{2(k+1)}(v_{i})$, $i\in I$, are disjoint from $V_{1}$.
		 Observe that
		$$|I|\geq \ell - |E(G')|/(2k+1)^{4(k+1)+2} \geq |E(G')|/(2k+1)^{4(k+1)+2}.$$
			
		\begin{claim} In time $\Oh(|E(G')|)$ we can either conclude that $\ctw(G')>k$, or for each $i\in I$ find a cycle in $G'$ passing only through the vertices of the ball $B_{2(k+1)}(v_i)$.
		\end{claim}
		\begin{proof}
		Suppose for some $i\in I$, $B_{2(k+1)}(v_i)$ does not contain a cycle.
		We will prove that every vertex in $G'[B_{2(k+1)}(v_i)]$ has degree at least 3 in $G'$, and that every edge appears with multiplicity~1.
		Notice first that every edge of the graph $G'[B_{2(k+1)}(v_{i})]$ has multiplicity 1, as otherwise an edge with multiplicity at least 2 would form a cycle,
		a contradiction.	Notice also that $B_{2(k+1)}(v_{i})$ does not have any vertex that has degree 2 in $G$.
		Indeed, recall that by the construction of the graph $G'$ any vertex of degree 2 is incident only to one edge of multiplicity 2, which is again a
		contradiction. Moreover, by the choice of $i\in [I]$, we obtain that $B_{2(k+1)}(v_{i})\cap V_{1}=\emptyset$ and therefore, $G'[B_{2(k+1)}(v_{i})]$ 
		does not have any vertex that has degree 1 in $G$.
		We conclude that every vertex in $G'[B_{2(k+1)}(v_i)]$ has degree at least 3 in $G$, and every edge appears with multiplicity~1.
		Recall that the subgraph of $G'$ induced by $B_{2(k+1)}(v_i)$ contains the full breadth-first search tree of vertices at distance at most $2(k+1)$ from $v_{i}$.
		If $G'[B_{2(k+1)}(v_i)]$ did not contain any cycle, then it would be equal to this breadth-first search tree, and in this tree all vertices except possibly the last layer would have degrees at least 3. 
		Hence, $G'$ would contain as a subgraph a perfect binary tree of height $2(k+1)$.
		From Lemma~\ref{lem:ekjgherk}, this tree has cutwidth at least $k+1$.
		The algorithm can thus check (by breadth-first search) for a cycle in the subgraph induced by $B_{2(k+1)}(v_i)$. 
		If it does not find any such cycle it immediately concludes that $\ctw(G)=\ctw(G') > k$.
		
		If for every $i\in I$, the breadth-first search in $G'[B_{2(k+1)}(v_{i})]$ finds a cycle, then the algorithm obtained, in total time $\Oh(|E(G')|)$, a set of at least 
		$|I|\geq E(G')/(2(k+1))^{4(k+1)+2}$ vertex-disjoint (and hence edge-disjoint) cycles.
		\cqed\end{proof}
		
		Let us assume that the algorithm has now found a set ${\cal C}$ of at least $E(G')/(2k+1)^{4(k+1)+2}$ edge-disjoint cycles and 
		let $H$ be the subgraph obtained from $G'$ by removing one, arbitrarily chosen, edge $e_C$ from each cycle $C\in {\cal C}$.
		Then $H$ can be immersed in $G'$ and $|E(H)| \leq |E(G')| \cdot (1-1/(2k+1)^{4(k+1)+2})$. 
		To complete the proof of the lemma we will prove that if $\sigma$ is any ordering of the vertices of $H$ then $\sigma$ is also an ordering of the vertices of $G'$ such that
		$\ctw_{\sigma}(G')\leq 2\ctw_{\sigma}(H)$.
		Notice that by reintroducing an edge $e_C$ of $G'$ to $H$ we increase the width of the $\sigma$-cuts separating its endpoints by exactly 1.
		Observe also that since $e_C$ belongs to the cycle $C$, the rest of the cycle forms a path $P_C$ in $H$ that connects the endpoints of $e_C$. 
		Therefore, each of the $\sigma$-cuts separating the endpoints of $e_C$ has to contain at least one edge of $P_C$. 
		Since for different edges $e_C$, for $C\in \mathcal{C}$, the corresponding paths $P_C$ are pairwise edge-disjoint and they are present in $H$, 
		it follows that the size of each $\sigma$-cut in $G'$ is at most twice the size of this $\sigma$-cut in $H$.
		Therefore $\ctw_{\sigma}(G')\leq 2 \ctw_{\sigma}(H)$.
		Thus, $H$ can be returned, concluding the algorithm.
\end{proof}

We are now ready to put all the pieces together.

\begin{proof}[Proof of Theorem~\ref{thm:algo}]
Given an $n$-vertex graph $G$ and an integer $k$, one can in time $2^{\Oh(k^2\log k)}\cdot n$ either conclude that $\ctw(G)>k$,
or output an ordering of $G$ of width at most $k$.
The proof follows the same recursive Reduction\&Compression scheme as the algorithm of Bodlaender~\cite{Bodlaender96}. 
By applying Lemma~\ref{lem:reduce}, we obtain a significantly smaller immersion $H$, and we recurse on $H$. This recursive call either concludes that $\ctw(H)>k$, which implies $\ctw(G)>k$,
or it produces an ordering of $H$ of optimum width $\ctw(H)\leq k$. This ordering can be lifted,
using Lemma~\ref{lem:reduce} again, to an ordering of $G$ of width $\leq 2k$. Given this ordering, we apply the dynamic programming procedure of Lemma~\ref{lem:dynamicCtw} to
construct an optimum ordering of $G$ in time $2^{\Oh(k^2\log k)}\cdot |V(G)|$.

Since at each recursion step the number of edges of the graph drops by a multiplicative factor of at least $1/(2k+1)^{4(k+1)+3}$, we see that the graph $G_i$ at level $i$ of the recursion will have at most
$(1-1/(2k+1)^{4(k+1)+3})^i\cdot |E(G)|$ edges. Hence, the total work used by the algorithm is bounded by the sum of a geometric series:
\begin{eqnarray}
\sum_{i=0}^{\infty}\, 2^{\Oh(k^2 \log k)} \cdot |E(G_{i})| &  \leq & 2^{\Oh(k^2 \log k)} \cdot |E(G)| \cdot \sum_{i=0}^{\infty}\, (1-1/(2k+1)^{4k+7})^{i} \nonumber\\
& = &  2^{\Oh(k^2 \log k)} \cdot |E(G)| \cdot (2k+1)^{4k+7} \\
& = & 2^{\Oh(k^2 \log k)} \cdot |E(G)|.\nonumber
\end{eqnarray}
\end{proof}

\pagebreak

\section{Obstructions to edge-removal distance to cutwidth }\label{sec:remddist}

Throughout this section, by $\Oh_k(w)$ we mean a quantity bounded by $c_k\cdot w+d_k$, for some constant $c_k,d_k$ depending on $k$ only.

Given a graph $G$ and a $k\in\NN$, we define 
the parameter $\dcw_k(G)$
as the minimum number of edges that can be deleted from 
$G$ so that the resulting graph has cutwidth at most $k$ (so $\dcw_k(G)$ fits in the wider category of ``graph  modification parameters'').
In other words:
$$\dcw_k(G)=\min\{ |F| \colon F\subseteq E(G) \mbox{~and~} \ctw(G\setminus F)\leq k\}$$
Let ${\cal C}_{w,k}=\{G\mid \dcw_k(G)\leq w\}$. Notice that ${\cal C}_{k}={\cal C}_{0,k}$.

In this section, we provide bounds to the sizes of the obstruction sets  
of the class of graphs $G$ with $\dcw_k(G)\leq w$, for each $k,w\in\NN$.
Our results are the following.

\begin{theorem}
\label{mainlin}
For every $w,k\in\NN$,
every graph in  ${\bf obs}_{\leq_{\rm si}}({\cal C}_{w,k})$
has $\Oh_k(w)$ vertices.
\end{theorem}

\begin{theorem}
\label{maisnlilowern}
For every $k,w\in\NN$ where $k\geq 7$,
the set ${\bf obs}_{\leq_{\rm i}}({\cal C}_{w,k})$
contains at least $\binom{3^{k-7} + w + 1}{w+1}$
non-isomorphic graphs. 
\end{theorem}
From  Observation~\ref{osesimme}, both bounds of
 Theorems~\ref{mainlin} and~\ref{maisnlilowern} holds  for both immersions and strong immersions as well. 

Given a collection ${\cal H}$ of graphs, we define 
the parameter $\aic_{\cal H}(G)$
as the minimum number of edges whose removal from $G$
creates an ${\cal H}$-immersion free graph, that is, a graph that does not admit any graph from ${\cal H}$ as an immersion. 
In both subsections that follow, we need the following observation.

\begin{observation}
\label{obpot}
For every graph $G$ and every  $w\in \NN$, it holds that $\dcw_w(G)=\aic_{{\rm obs}({\cal C}_{w})}(G)$.
\end{observation}

We remark that, within the same set of authors, we have recently studied kernelization algorithms for edge removal problems to immersion-closed classes.
The following result has been obtained in~\cite{GiannopoulouPTRW2016}: 
whenever a finite collection of graphs ${\cal H}$ contains at least one planar subcubic graph, and all graphs from ${\cal H}$ are connected, then the problem
of computing $\aic_{\cal H}(G)$, parameterized by the target value, admits a linear kernel.
These prerequisites are satisfied for ${\cal H}={\rm obs}({\cal C}_{w})$, and hence the problem of computing $\dcw_w(G)$, parameterized by the target value $k$, admits a linear kernel.

The connections between kernelization procedures and upper bounds on minimal obstruction sizes have already been noticed in the literature; see e.g.~\cite{FominLMS12}. 
Intuitively, whenever the kernelization rules apply only minor or immersion operations, the kernelization algorithm can be turned 
into a proof of an upper bound on the sizes of minimal obstacles for the corresponding order.
Unfortunately, this is not the case for the results of~\cite{GiannopoulouPTRW2016}: the main problem is the lack of linked decompositions for parameter tree-cut width, which plays the central role.
Here, the situation is different, as we know that there are always linked orderings of optimum width.
We therefore showcase how to use the linkedness to obtain a linear upper bound on the sizes of obstructions for ${\cal C}_{w,k}$.
The arguments are somewhat similar as in~\cite{GiannopoulouPTRW2016}: we use the idea of protrusions, adapted to the setting of edge cuts, and we try to replace protrusions with smaller ones having the same behavior.
The main point is that linkedness ensures us that the replacement results in an immersion of the original graph.

\subsection{Upper bound on obstruction size}


A \emph{partial $q$-boundaried graph} is a pair $\mathbf{G}=(G,\bar{x})$ where $G$ is a graph and $\bar{x}=(x_{1},\dots,x_q)$ 
is a $q$-tuple that consists either of vertices of $G$ or from empty slots (that are indices that do not correspond to vertices of $G$). If $x_{i}$ is an empty slot, we denote it by $x_{i}=\diamond$. 
The extension of such $\mathbf{G}$ is defined just as for $q$-boundaried graphs, but we put
$x'_{i}=\diamond'$ iff $x_{i}=\diamond$.
Intuitively a partial $q$-boundaried graph extends the notion of a boundaried graph by allowing the vertices of the boundary to carry 
indices from a set whose cardinality might be bigger than the boundary.

Let $H$ be a graph and let  $(X_{1},X_{2})$ be its cut where $q=\delta(X_{1})$.
Let $E_H(X_1,X_2)=\{e_{1},\ldots,e_{q}\}$ where $e_{i}=\{x_{i}^{1},x_{i}^{2}\}$, $i\in[q]$, and such that 
$x_i^j\in X_{j}$ for $(i,j)\in[q]\times[2]$. 
For $j\in[2]$, we 
say that the pair $(X_1,X_2)$ {\em{generates}} 
the $q$-boundaried graph ${\bf A}_j=({A}_j,\overline{x}_{j})$ if $A_i=G[X_i]$
and $\overline{x}_{i}=(x_{1}^{j},\ldots,x_{q}^{j})$. 

We denote by ${\cal B}_{q,h}$ the collection containing  every $q$-boundaried graph that can be generated from some cut   $(X_{1},X_{2})$ of some graph $H$ where $|V(H)|+|E(H)|\leq h$ and $q=\delta(X_{1})$.
Moreover, we denote by ${\cal M}_{q,h}$ the set of all partial $q$-boundaried graphs $\mathbf{F}'=(F',\bar{x}')$
 such that, for some $\mathbf{F} =(F,\bar{x})$ of $\mathcal{B}_{q,h}$, $F'$ is a subgraph of $F$ and a vertex $x_i$ in $\bar{x}'$ is an empty slot iff  
 $x_{i}\in V(F)\setminus V(F')$.

In other words, ${\cal M}_{q,h}$ contains all partial $q$-boundaried graphs that can be 
generated by a graph whose number of edges and vertices does not exceed $h$.  
We insist that  if ${\bf H}=(H,\overline{x})\in{\cal M}_{q,h}$, then the 
vertices of $H$ are taken from some fixed repository of $h$ vertices and 
that an element $x_{i}$ of $\overline{x}$ is either an empty slot (i.e., $x_{i}=\diamond$) or 
the $i$-th vertex of some predetermined ordering $(x_{1},\ldots,x_{q})$ of $q$ vertices 
from this repository. This permits us to assume that $|{\cal M}_{q,h}|$ is bounded by some function that depends only on $q$ and $h$.

Let  ${\bf G}=(G,\overline{x})$ be a $q$-boundaried graph
and ${\bf H}=(H,\overline{y})$ be  a partial $q$-boundaried graph. Let also   
$G^*$ and $H^*$ be the extensions of $\mathbf{G}$ and $\mathbf{H}$, respectively.  
We also assume  that, for all $i\in[q]$, either $y_i=x_i$ or $y_i=\diamond$. For an edge subset $R\subseteq E(G^*)$,  
we say that 
${\bf H}$ is an {\em $R$-avoiding strong immersion in ${\bf G}$}  if there is an $H^*$-immersion model $(\phi,\psi)$ of $G^*\setminus R$ where, for every $i\in[q]$ such that $y_{i}\neq\diamond$, 
it holds that $\phi(y_{i})=x_{i}$ 
and 
$\phi(y_{i}^{\prime})=x_{i}'\neq\diamond$. We now define the {\em $R$-avoiding $(q,h)$-folio} of ${\bf G}$ as the set of all partial $q$-boundaried graphs in ${\cal M}_{q,h}$ 
that are $R$-avoiding strong immersions in ${\bf G}$ and we denote it by ${\bf folio}_{q,h,R}({\bf G})$. 
We finally define 
$${\cal F}_{q,h}({\bf G})=\{{\bf folio}_{q,h,R}({\bf G})\mid R\subseteq E(G^*)\ \text{and}\ |R|\leq q\}.$$
Given two  $q$-boundaried graphs  ${\bf G}_{1}$ and ${\bf G}_{2}$ 
we write ${\bf G}_{1}\sim_{q,h}{\bf G}_{2}$ in order to denote that   
${\cal F}_{q,h}({\bf G}_{2})={\cal F}_{q,h}({\bf G}_{2})$.
 As ${\cal F}_{q,h}$
maps each $q$-boundaried graph to a collection of subsets of ${\cal M}_{q,h}$ we have the following.

\begin{lemma}
\label{k8io90op}
There is some function $f_1:\NN^2\rightarrow \NN$
such that for every two  non-negative integers $q$ and $h$, the 
number of equivalence classes of $\sim_{q,h}$ is at most 
$f_1(q,h)$.
\end{lemma}

The next lemma is a consequence of the definition of the function ${\cal F}_{q,h}$.

\begin{lemma}\label{euitq}
Let ${\cal H}$ be some set of connected graphs, each of  at most $h$ vertices, and let ${\bf G}_i=(G_i,\overline{x}_{i}), i\in\{1,2\}$ be two $q$-boundaried graphs such that ${\bf G}_{1}\sim_{q,h}{\bf G}_{2}$
and such both $G_{1},G_{2}$ are  ${\cal H}$-immersion free.
Then, for every $q$-boundaried graph ${\bf F}=(F,\overline{y})$, it holds that $\aic_{\cal H}({\bf F}\oplus {\bf G}_{1})=\aic_{\cal H}({\bf F}\oplus {\bf G}_{2})$.
\end{lemma}

The proof is omitted as it is very similar to the one in~\cite{Chatzidimitriou15} where a similar encoding 
was defined in order to treat the topological minor relation. 
To see the main idea, recall 
that  ${\cal F}_{q,h}({\bf G}_{i})$ registers all different ``partial occurrences'' of 
graphs of $\leq h$ vertices (and therefore also of graphs of ${\cal H}$)
in ${\bf G}_i'$, for all possible ways to obtain ${\bf G}_i'$ from ${\bf G}_i$ after removing at most $q$ edges. This encoding is indeed sufficient to capture the behavior 
of all possible edge sets whose removal from ${\bf F}\oplus{\bf G}_{i}$ creates an ${\cal H}$-free graph. Indeed, as both $G_{1}$ and $G_{2}$ are  ${\cal H}$-immersion free, any such 
set should have at most $q$ edges inside ${\bf G}_{i}$ as, if not, the $q$-boundary edges 
between ${\bf F}$ and ${\bf G}_i$ would also make the same job. A similar discussion is also present in~\cite{GiannopoulouPTRW2016}.

Given a graph $G$ and $X\subseteq V(G)$, we write ${\bf cw}_{\sigma}(G,X)=\delta_G(X)+{\bf cw}_{\sigma_{X}}(G[X])$. We require the following extension of the definition of linked orderings.

\begin{definition}[extended linked ordering]
Let $G$ be a $n$-vertex graph and $X\subseteq V(G)$.
An ordering $\sigma=\langle v_{1},\ldots,v_{n}\rangle$ of $G$ is $X$-\emph{linked} if  $X=\{v_{n-|X|+1},\ldots,v_{n}\}$
and 
for every $i,j\in[n-|X|,n]$ where $i<j$
there exist $\min \{\delta(\{{v_{1},\ldots,v_h}\}) \mid  i\leq h\leq j \}$ edge-disjoint paths between $\{v_{1},\ldots,v_{i}\}$ and $\{v_{j},\ldots,v_{n}\}$ in $G$.
\end{definition}

The proof of the following result is very similar to the one of Lemma~\ref{lem:linked}. We just move $X$ to the end of the ordering, in the order given by $\sigma$,
and apply exhaustively the same refinement step based on submodularity, but only to the subordering induced by $X$.

\begin{lemma}
\label{linkedhelp}
For every graph $G$ and every subset $X$ 
of $V(G)$, if there exists an ordering $\sigma$ of $G$
such that ${\bf cw}_{\sigma}(G,X)\leq r$, then 
there exists  an $X$-linked ordering $\sigma'$ of $G$
such that ${\bf cw}_{\sigma'}(G,X)\leq r$.
\end{lemma}

Let $w_1,w_2\in\NN$, $G$ be a graph, and $X\subseteq V(G)$. We say that $X$ is an {\em $(w_1,w_2)$-cutwidth-edge-protrusion} of $G$ 
if $\delta(X)\leq w_1$ and ${\bf cw}(G[X_{i}])\leq w_2$.

The next lemma uses an idea similar to the one of Lemma~\ref{restated}. Here $\sim_{q,h}$ plays the role of  $(q,\ell)$-similarity.

\begin{lemma}
\label{wmopri9o}
There is a computable function $f_2:\NN^2\rightarrow\NN$ such that the following holds:
Let $k$ be a non-negative integer and let ${\cal H}$ be a 
finite set of connected graphs, each having at most $h$ vertices and edges. 
Let also $G$ be a graph and let $X$ be a $(2k,k)$-cutwidth-edge-protrusion of $G$.
If $|X|>f_2(k,h)$, then $G$ contains as a proper strong immersion 
a graph $G'$ where $\aic_{\cal H}(G)=\aic_{\cal H}(G')$.
\end{lemma}

\begin{proof}
We set $f_{2}(k,h)=(f_{1}(3k,h)+1)^{3k+1}-1$.
We have that $|X|>f_{2}(k,h)$ or, equivalently, $|X|\geq (f_{1}(3k,h)+1)^{3k+1}$. We set $\ell=|X|$.
Let $\sigma^*=\langle x_{1},\ldots,x_{\ell}\rangle$ be an ordering of the vertices in $X$
such that ${\bf cw}_{\sigma^*}(G[X])\leq k$. 
Let $\sigma=\langle v_{1},\ldots,v_{n-\ell},v_{n-\ell+1},\ldots,v_{n}\rangle$
be any ordering of $V(G)$ such that $\sigma^*$ is a suffix of $\sigma$, i.e. $\langle x_{1},\ldots,x_{\ell}\rangle=\langle v_{n-\ell+1},\ldots,v_{n}\rangle$. 
It follows that ${\bf cw}_{\sigma}(G,X)\leq {\bf cw}_{\sigma^*}(G[X])+\delta_{G}(X)\leq k+2k=3k$. 
From Lemma~\ref{linkedhelp}, there is an $X$-linked ordering $\sigma'$ of $V(G)$, where ${\bf cw}_{\sigma'}(G,X)\leq 3k$.

We set $k_i=\delta_{G[X]}(\{x_{1},\ldots,x_{i-(n-\ell)})+\delta_{G}(X)$
and observe that $k_{i}\leq k+2k=3k,  i\in[n-\ell,n-1]$.
We set up the alphabet $\mathbb{A}=\{0,1,\ldots,3k\}$ 
and we see $w=k_{n-\ell},k_{n-\ell+1},\ldots,k_{n-1}$
as a word on $\mathbb{A}$. Let also $N=f_{1}(3k,h)$.
Notice that $|w|=|X|=\ell\geq (f_{1}(3k,h)+1)^{3k+1}=(N+1)^{|\mathbb{A}|}$.
From Corollary~\ref{lem:wordlngth}, if 
$|w|\geq (N+1)^{|\mathbb{A}|}$,
there are $a,b\in [n-\ell,n-1], a<b$
and some $p\in\mathbb{A}$ 
such that $k_{a},k_{b}\geq p$
and $p$ appears in $\{k_{a},\ldots,k_{b}\}$
at least $N+1$ times.
Let these appearance be at indices $a\leq i_1<i_2<\ldots<i_{N+1}\leq b$.

By $X$-linkedness, there are $p$ edge-disjoint paths $P^i$, for $i\in [p]$, from $\{v_1,\ldots,v_{i_1}\}$ to $\{v_{i_{N+1}},\ldots,$
$v_{|V(G)|}\}$.
Observe that for each $j\in [N+1]$, each path $P^i$ must cross exactly one edge of the cut $\delta_G(\{v_1,\ldots,v_{i_j}\})$; let this edge be $z^i_jw^i_j$, where 
$z^i_j\in \{v_1,\ldots,v_{i_j}\}$ and $w^i_j\notin \{v_1,\ldots,v_{i_j}\}$.
For each $j\in [N+1]$ we define $p$-boundaried graphs
${\bf F}_j=(F_j,(z^1_j,\ldots,z^p_j))$, where  
$F_{j}=G[\{v_{1},\ldots,v_{i_j}\}]$, and
${\bf G}_{j}=(G_{j},(w^1_j,\ldots,w^p_j))$
where $G_{j}=G[\{v_{i_j+1},\ldots,v_{|V(G)|}\}]$.

As, from Lemma~\ref{k8io90op},
the equivalence relation $\sim_{3k,h}$ has at most $N$ equivalent classes,
there are $j_1,j_2$ such that $a\leq j_1<j_2\leq b$
such that
${\bf G}_{j_1}\sim_{3k,h} {\bf G}_{j_2}$. Let $G'={\bf F}_{i_{j_1}}\oplus {\bf G}_{i_{j_2}}$; it is easy to observe that $G'$ is a proper immersion of $G$, because the edges added when joining can be modeled using
appropriate infixes of the paths $P_i$.
From Lemma~\ref{euitq}, however, we have  $\aic_{\cal H}({\bf F}_{i}\oplus {\bf G}_{i})=\aic_{\cal H}({\bf F}_{i}\oplus {\bf G}_{j})$, and therefore 
$\aic_{\cal H}(G)=\aic_{\cal H}(G')$.
\end{proof}

 \begin{lemma}
 \label{mroklp3}
 Let $k,w,\ell\in\NN$ and let $G$ be a graph.  If $\dcw_k(G)\leq w$
 and $|V(G)|\geq \ell\cdot (2w+1)+2w$, 
 then $G$ has  a $(2k,k)$-cutwidth-edge-protrusion $X$
 where $|X|\geq \ell$.
\end{lemma}
\begin{proof}
We denote $n=|V(G)|$.
Let $F$ be a set of edges of $G$ such that if $G'=G\setminus F$, then ${\bf cw}(G')\leq k$.
Let $\sigma=\langle v_{1},\ldots,v_{n}\rangle$ be an ordering of $V(G')$
such that ${\bf cw}_{\sigma}(G')\leq k$.
Let $I$ denote the indices in $\sigma$ of the endpoints of the edges
in~$F$ and notice that $|I| \leq 2w$. We consider the maximal
intervals of $[n]$ that do not intersect~$I$.
The set $[n] \setminus I$ has $n-|I|\geq n-2w$ elements, that are distributed
among at most $|I|+1\leq 2w+1$ such intervals. By the pigeonhole
principle, one interval $\{i, \dots, j\}$ has
at least $\frac{n - 2w}{2w+1} \geq \ell$ elements.

Consider now the set $X=\{v_{i},\ldots,v_{j}\}$, $|X|\geq \ell$.
Notice that if $\sigma'=\langle v_{i},\ldots,v_{j} \rangle$,
then ${\bf cw}_{\sigma'}(G[X])\leq k$. Moreover 
there are at most
$\delta_{G'}(\{v_{1},\ldots,v_{i-1}\})+\delta_{G'}(\{v_{j+1},\ldots,v_{n}\})\leq
2k$ edges with one vertex in $X$ and the other not in $X$. Therefore,
$\delta_{G}(X)\leq 2k$ and $X$ is 
a $(2k,k)$-cutwidth-edge-protrusion of $G$. 
%
\end{proof}

\begin{proof}[Proof of Theorem~\ref{mainlin}]
We set ${\cal H}={\bf obs}_{\leq_{\rm im}}({\cal C}_{k})$.
By Theorem~\ref{restated}, there is a function $f_{3}:\NN\rightarrow\NN$ such that graphs from ${\cal H}$ have at most $h=f_{3}(k)$ vertices. 
Let $G\in{\bf obs}_{\leq_{\rm im}}({\cal C}_{w,k})$.
This means that $\dcw_k(G)=w+1$, while, 
for every proper strong immersion $G'$ of $G$, it holds that 
$\dcw_k(G')\leq w$.
This, together with Observation~\ref{obpot} and Lemma~\ref{wmopri9o}, implies 
that $G$ cannot have 
a $(2k,k)$-cutwidth-edge-protrusion $X$ of more than $\ell=f_{2}(k,h)$, vertices. As  $\dcw_k(G)=w+1$, Lemma~\ref{mroklp3} implies that $|V(G)|< \ell\cdot (2w+3)+2(w+1)=\Oh_{k}(w)$
vertices.
\end{proof}

Notice that Theorem~\ref{mainlin} can be seen as an application of Lemma~\ref{lem:linked} on the existence of linked orderings of optimum cutwidth (along with Lemma~\ref{linkedhelp}, that is an easy extension of it). We stress that this bound is constructive as, by going through the proof, one can  make an estimation of the functions of $k$ hidden in the $\Oh_{k}$ notation. In future work we plan to further develop  the above technique of proving such linear bounds for other edge modification problems. Some analogous work for vertex modification problems have been done in~\cite{FominLMS12} where the parameter $k$ is the number of vertices that one should remove in order to transform a graph to 
one of treewidth at most $k$. The corresponding bound 
in~\cite{FominLMS12} is $w^{\Oh_{k}(1)}$, is non-constructive, and follows a distinct (more elaborated) approach.

\subsection{Lower bound on number of obstructions}

We now focus on the proof of Theorem~\ref{maisnlilowern}.
We need the following result.
\begin{theorem}[\!\!\cite{GOVINDAN2001189}]\label{t:gov}
  For every $k \geq 7$, the number of non-isomorphic connected minimal obstructions in
  ${\bf obs}_{\leq_{\rm i}}({\cal C}_{k})$ is at least $3^{k-7} + 1$.
\end{theorem}

Recall that, given a graph class ${\cal H}$, we defined $\aic_{\cal H}(G)$ as  the minimum number of edges
of $G$ whose removal yields an ${\cal H}$-immersion-free graph. 
We set ${\cal C}_{w,{\cal H}}=\{G\mid \aic_{\cal H}(G)\leq w\}$. In particular ${\cal C}_{0,{\cal H}}$ is the class of all 
${\cal H}$-immersion free graphs. If $G$ and $H$ are graphs, we denote by $G\uplus H$ the disjoint union of $G$ and $H$.

The following observations follow directly from the definition of $\aic_{\cal H}$.

\begin{observation}
\label{obskl}
If $G$ and $H$ are graphs,  then $H\leq_{\rm i}G$ implies that $\aic_{\cal H}(H)\leq \aic_{\cal H}(G)$. 
\end{observation}

\begin{observation}
\label{obskl33}
If $G$ and $H$ are graphs,  then $\aic_{\cal H}(G\uplus H) = \aic_{\cal H}(G) + \aic_{\cal H}(H)$.
\end{observation}

\begin{observation}
\label{obs0o9}
If $G\in{\bf obs}_{\leq_{\rm i}}({\cal C}_{w,{\cal H}})$, then $\aic_{\cal H}(G)=w+1$.
\end{observation}

%

%

\begin{lemma}\label{c:zero}
Let ${\cal H}$ be some  $\leq_{\rm i}$-antichain.
  For every non-negative integer $w$, if $G_1, \dots,G_{w+1}$ are (not necessarily
  distinct) members of  ${\cal H}$, then $\biguplus_{i=1}^{w+1} G_i\in{\bf obs}_{\leq_{\rm i}}({\cal C}_{w,{\cal H}})$.
\end{lemma}

\begin{proof}
  Let $G = \biguplus_{i=1}^{w+1} G_i$.  To prove that 
 $G\in {\bf obs}_{\leq_{\rm i}}({\cal C}_{w,{\cal H}})$ we have to show that it satisfies {\bf O1} and {\bf O2}.
Notice that since $\mathcal{H}$ is an $\leq_{\rm i}$-antichain, $\aic_{\cal H}(H) = 1$ for every $H \in \mathcal{H}$.
By Observations~\ref{obskl33} and~\ref{obs0o9}, $\aic_{\cal H}(G) = \sum_{i=1}^{w+1}\aic_{\cal H}(G_i) = w+1$ and {\bf O1} holds. Therefore, $G\not\in {\cal C}_{w,{\cal H}}$.
Let now $G'$ is a proper immersion of $G$.
This mean that $G'=\biguplus_{i=1}^{w+1}G'_{i}$
where $G'_{i}\leq_{\rm i}G_{i}$ and at least one of $G_{1}',\ldots,G_{w+1}'$ is different than $G_{i}$. W.l.o.g. we assume that this graph is $G_{w+1}$.
As ${\cal H}$ is a $\leq_{\rm i}$-antichain, $G'_{w+1}$
is not isomorphic to a graph of ${\cal H}$. Therefore $\aic_{\cal H}(G'_{w+1})=0$. 
Then, by Observations~\ref{obskl} and~\ref{obskl33}, $\aic_{\cal H}(G')=\sum_{i=1}^{k}\aic_{\cal H}(G'_{i})  +\aic_{\cal H}(G'_{w+1})\leq \sum_{i=1}^{w}\aic_{\cal H}(G_{i})  +0=w$ and {\bf O2} holds.
%
%
%
\end{proof}

\begin{theorem}
\label{u7ui87u}
If $k$ is a non-negative integer and  ${\cal H}$
is a $\leq_{\rm i}$-antichain that contains  at least $q$ connected graphs, then $|{\bf obs}_{\leq_{\rm i}}({\cal C}_{w,{\cal H}})|\geq \binom{ q+ w}{w+1}$.
\end{theorem}

\begin{proof}
Let ${\cal H}'$ be some subset of ${\cal H}$ containing $q$ connected graphs.   
Using Lemma~\ref{c:zero}, we observe that every multiset
of cardinality $w+1$  whose
elements belong to ${\cal H}'$ corresponds to a different
(i.e.\ non-isomorphic)
obstruction of ${\cal C}_{w,{\cal H}}$.  
Therefore, $|{\bf obs}_{\leq_{\rm i}}({\cal C}_{w,{\cal H}})|$ is  at least the number of multisets of cardinality
$w+1$ the elements of which are taken from a set of cardinality $q$, which is
known to be~$\binom{q+w}{w+1}$.
\end{proof}
%
%

\begin{proof}[Proof of Theorem~\ref{maisnlilowern}]
From Observation~\ref{obpot}, ${\cal C}_{w,k}={\cal C}_{w,{\cal H}_{k}}$, where ${\cal H}_{k}={\bf obs}_{\leq_{\rm i}}({\cal C}_{k})$.
This means that ${\bf obs}_{\leq_{\rm i}}({\cal C}_{w,k})={\bf obs}_{\leq_{\rm i}}({\cal C}_{w,{\cal H}_k}).$ The result follows from Theorems~\ref{t:gov} and \ref{u7ui87u}.
\end{proof}

\section{Conclusions}\label{sec:conc}

In this paper we have proved that the immersion obstructions for admitting a  layout of cutwidth at most $k$ have sizes single-exponential in $\Oh(k^3\log k)$.
The core of the proof can be interpreted as bounding the number of different behavior types for a part of the graph that has only a small number of edges connecting it to the rest.
This, in turn, gives an upper bound on the number of states for a dynamic programming algorithm that computes the optimum cutwidth ordering on an approximate one.
This last result, complemented with an adaptation of the reduction scheme of Bodlaender~\cite{Bodlaender96} to the setting of cutwidth, yields a direct and self-contained FPT algorithm
for computing the cutwidth of a graph. In fact, we believe that our algorithm can be thought of ``Bodlaender's algorithm for treewidth in a nutshell''. It consists of the same
two components, namely a recursive reduction scheme and dynamic programming on an approximate decomposition, but the less challenging setting of cutwidth makes both components
simpler, thus making the key ideas easier to understand.
For an alternative attempt of simplification of the 
algorithm of Bodlaender and Kloks~\cite{BodlaenderK96}, applied for the case of pathwidth, see~\cite{Furer2016}.

In our proof of the upper bound on the number of types/states, we used a somewhat new bucketing approach. This approach holds the essence of the typical sequences of Bodlaender and Kloks~\cite{BodlaenderK96},
but we find it more natural and conceptually simpler. The drawback is that we lose a $\log k$ factor in the exponent. It is conceivable that we could refine
our results by removing this factor provided we applied typical sequences directly, but this is a price that we are willing to pay for the sake of simplicity and being self-contained.

An important ingredient of our approach is the observation that there is always an optimum cutwidth ordering that is linked: 
the cutsizes along the ordering precisely govern the edge connectivity between prefixes and suffixes.
Recently, there is a growing interest in parameters that are tree-like analogues of cutwidth: tree-cut width~\cite{Wollan15} and carving-width~\cite{SeymourT94}.
In future work, we aim to explore and use linkedness for tree-cut decompositions and carving decompositions in a similar manner as presented here. 

\subsection{Acknowledgements.} The second author thanks Mikołaj Bojańczyk for the common work on understanding and reinterpreting the Bodlaender-Kloks dynamic programming algorithm~\cite{BodlaenderK96},
which influenced the bucketing approach presented in this paper.
We also thank O-joung Kwon for pointing us to~\cite{GeelenGW02,KanteK14}, as well as an anonymous referee for noting that the running time in Lemma~\ref{lem:reduce} can be reduced to polynomial by amortization.

\pagebreak

\bibliographystyle{plainurl}
\bibliography{cutwidth-obs}

\end{document}